\documentclass[lettersize,journal]{IEEEtran}
\usepackage{mathtools}
\usepackage{amsmath}
\usepackage[shortlabels]{enumitem}

\usepackage{graphicx,epic,eepic,epsfig,latexsym,verbatim,color,amssymb}
 
\usepackage{amsfonts}       
\usepackage{nicefrac}       
\usepackage[linesnumbered,ruled,procnumbered]{algorithm2e}
\usepackage[table]{xcolor}

\usepackage{framed}
\definecolor{shadecolor}{rgb}{0.9,0.90,0.9}
\usepackage{bbm}
 
\usepackage{float}
\usepackage{tikz}
\usetikzlibrary{chains}
\usetikzlibrary{fit}

\usepackage{epsfig}
\usetikzlibrary{shapes.symbols,patterns} 
\usepackage{pgfplots}
\pgfplotsset{compat=1.18}

\usepackage[strict]{changepage}
\usepackage{hyperref}
\hypersetup{colorlinks=true,citecolor=blue,linkcolor=blue,filecolor=blue,urlcolor=blue,breaklinks=true}

\usepackage[marginal]{footmisc}
\usepackage{url}
\usepackage{theorem}

\newtheorem{definition}{Definition}
\newtheorem{proposition}{Proposition}

\newtheorem{theorem}[proposition]{Theorem}

\def\squareforqed{\hbox{\rlap{$\sqcap$}$\sqcup$}}
\def\qed{\ifmmode\squareforqed\else{\unskip\nobreak\hfil
\penalty50\hskip1em\null\nobreak\hfil\squareforqed
\parfillskip=0pt\finalhyphendemerits=0\endgraf}\fi}
\def\endenv{\ifmmode\;\else{\unskip\nobreak\hfil
\penalty50\hskip1em\null\nobreak\hfil\;
\parfillskip=0pt\finalhyphendemerits=0\endgraf}\fi}
\newenvironment{proof}{\noindent \textbf{{Proof~} }}{\hfill $QED$}

\newcounter{remark}
\newenvironment{remark}[1][]{\refstepcounter{remark}\par\medskip\noindent%
\textbf{Remark~\theremark #1} }{\medskip}

\newcounter{example}

\mathchardef\ordinarycolon\mathcode`\:
\mathcode`\:=\string"8000
\def\vcentcolon{\mathrel{\mathop\ordinarycolon}}
\begingroup \catcode`\:=\active
  \lowercase{\endgroup
  \let :\vcentcolon
  }

\usepackage{cleveref}
\usepackage{graphicx}

\RequirePackage[framemethod=default]{mdframed}
\newmdenv[skipabove=7pt,
skipbelow=7pt,
backgroundcolor=darkblue!15,
innerleftmargin=5pt,
innerrightmargin=5pt,
innertopmargin=5pt,
leftmargin=0cm,
rightmargin=0cm,
innerbottommargin=5pt,
linewidth=1pt]{tBox}

\newmdenv[skipabove=7pt,
skipbelow=7pt,
backgroundcolor=red!15,
innerleftmargin=5pt,
innerrightmargin=5pt,
innertopmargin=5pt,
leftmargin=0cm,
rightmargin=0cm,
innerbottommargin=5pt,
linewidth=1pt]{rBox}

\newmdenv[skipabove=7pt,
skipbelow=7pt,
backgroundcolor=blue2!25,
innerleftmargin=5pt,
innerrightmargin=5pt,
innertopmargin=5pt,
leftmargin=0cm,
rightmargin=0cm,
innerbottommargin=5pt,
linewidth=1pt]{dBox}
\newmdenv[skipabove=7pt,
skipbelow=7pt,
backgroundcolor=darkkblue!15,
innerleftmargin=5pt,
innerrightmargin=5pt,
innertopmargin=5pt,
leftmargin=0cm,
rightmargin=0cm,
innerbottommargin=5pt,
linewidth=1pt]{sBox}
\definecolor{darkblue}{RGB}{0,76,156}
\definecolor{darkkblue}{RGB}{0,0,153}
\definecolor{blue2}{RGB}{102,178,255}
\definecolor{darkred}{RGB}{195,0,0}

\newcommand{\nc}{\newcommand}
\nc{\rnc}{\renewcommand}
\nc{\lbar}[1]{\overline{#1}}
\nc{\bra}[1]{\langle#1|}
\nc{\ket}[1]{|#1\rangle}
\nc{\dketbra}[2]{\vert #1 \rangle \hspace{-.8mm} \rangle \hspace{-.4mm} \langle\hspace{-.8mm}\langle #2 \vert}
\nc{\dbra}[1]{\langle\hspace{-.8mm}\langle #1\vert}
\nc{\dket}[1]{\vert#1\rangle\hspace{-.8mm}\rangle}
\nc{\ketbra}[2]{|#1\rangle\!\langle#2|}
\nc{\braket}[2]{\langle#1|#2\rangle}

\nc{\proj}[1]{| #1\rangle\!\langle #1 |}
\nc{\avg}[1]{\langle#1\rangle}
\nc{\rank}{\operatorname{Rank}}
\nc{\smfrac}[2]{\mbox{$\frac{#1}{#2}$}}
\nc{\tr}{\operatorname{Tr}}
\nc{\ox}{\otimes}
\nc{\dg}{\dagger}
\nc{\dn}{\downarrow}
\nc{\cA}{{\cal A}}
\nc{\cB}{{\cal B}}
\nc{\cC}{{\cal C}}
\nc{\cD}{{\cal D}}
\nc{\cE}{{\cal E}}
\nc{\cF}{{\cal F}}
\nc{\cG}{{\cal G}}
\nc{\cH}{{\cal H}}
\nc{\cI}{{\cal I}}
\nc{\cJ}{{\cal J}}
\nc{\cK}{{\cal K}}
\nc{\cL}{{\cal L}}
\nc{\cM}{{\cal M}}
\nc{\cN}{{\cal N}}
\nc{\cO}{{\cal O}}
\nc{\cP}{{\cal P}}
\nc{\cQ}{{\cal Q}}
\nc{\cR}{{\cal R}}
\nc{\cS}{{\cal S}}
\nc{\cT}{{\cal T}}
\nc{\cU}{{\cal U}}
\nc{\cV}{{\cal V}}
\nc{\cX}{{\cal X}}
\nc{\cY}{{\cal Y}}
\nc{\cZ}{{\cal Z}}
\nc{\cW}{{\cal W}}
\nc{\csupp}{{\operatorname{csupp}}}
\nc{\qsupp}{{\operatorname{qsupp}}}
\nc{\var}{{\operatorname{var}}}
\nc{\rar}{\rightarrow}
\nc{\lrar}{\longrightarrow}
\nc{\polylog}{{\operatorname{polylog}}}
\nc{\wt}{{\operatorname{wt}}}
\nc{\av}[1]{{\left\langle {#1} \right\rangle}}
\nc{\supp}{{\operatorname{supp}}}
\nc{\VComb}{{\widetilde{\cal C}}}
\nc{\VChoi}{{\widetilde{C}}}

\nc{\argmin}{{\operatorname{argmin}}}

\def\x{\xi}

\nc{\RR}{{{\mathbb R}}}
\nc{\CC}{{{\mathbb C}}}
\nc{\FF}{{{\mathbb F}}}
\nc{\NN}{{{\mathbb N}}}
\nc{\ZZ}{{{\mathbb Z}}}
\nc{\PP}{{{\mathbb P}}}
\nc{\QQ}{{{\mathbb Q}}}
\nc{\UU}{{{\mathbb U}}}
\nc{\EE}{{{\mathbb E}}}
\nc{\id}{{\operatorname{id}}}

\nc{\CHSH}{{\operatorname{CHSH}}}

\nc{\<}{\langle}
\rnc{\>}{\rangle}
\nc{\rU}{\mbox{U}}

\nc{\ob}[1]{#1}

\usepackage{tikz}
\usepackage{hyperref}
\hypersetup{colorlinks=true,citecolor=blue,linkcolor=blue,filecolor=blue,urlcolor=blue,breaklinks=true}

\makeatletter
\def\grd@save@target#1{%
  \def\grd@target{#1}}
\def\grd@save@start#1{%
  \def\grd@start{#1}}
\tikzset{
  grid with coordinates/.style={
    to path={%
      \pgfextra{%
        \edef\grd@@target{(\tikztotarget)}%
        \tikz@scan@one@point\grd@save@target\grd@@target\relax
        \edef\grd@@start{(\tikztostart)}%
        \tikz@scan@one@point\grd@save@start\grd@@start\relax
        \draw[minor help lines,magenta] (\tikztostart) grid (\tikztotarget);
        \draw[major help lines] (\tikztostart) grid (\tikztotarget);
        \grd@start
        \pgfmathsetmacro{\grd@xa}{\the\pgf@x/1cm}
        \pgfmathsetmacro{\grd@ya}{\the\pgf@y/1cm}
        \grd@target
        \pgfmathsetmacro{\grd@xb}{\the\pgf@x/1cm}
        \pgfmathsetmacro{\grd@yb}{\the\pgf@y/1cm}
        \pgfmathsetmacro{\grd@xc}{\grd@xa + \pgfkeysvalueof{/tikz/grid with coordinates/major step}}
        \pgfmathsetmacro{\grd@yc}{\grd@ya + \pgfkeysvalueof{/tikz/grid with coordinates/major step}}
        \foreach \x in {\grd@xa,\grd@xc,...,\grd@xb}
        \node[anchor=north] at (\x,\grd@ya) {\pgfmathprintnumber{\x}};
        \foreach \y in {\grd@ya,\grd@yc,...,\grd@yb}
        \node[anchor=east] at (\grd@xa,\y) {\pgfmathprintnumber{\y}};
      }
    }
  },
  minor help lines/.style={
    help lines,
    step=\pgfkeysvalueof{/tikz/grid with coordinates/minor step}
  },
  major help lines/.style={
    help lines,
    line width=\pgfkeysvalueof{/tikz/grid with coordinates/major line width},
    step=\pgfkeysvalueof{/tikz/grid with coordinates/major step}
  },
  grid with coordinates/.cd,
  minor step/.initial=.2,
  major step/.initial=1,
  major line width/.initial=2pt,
}
\makeatother

\usepackage{thmtools}
\usepackage{thm-restate}
\usepackage{etoolbox}
\makeatletter
\def\problem@s{}
\newcounter{problems@cnt}

\newcommand{\allproblems}{\problem@s}
\makeatother

\usepackage{tikz}
\usetikzlibrary{positioning}
\usetikzlibrary{shapes.geometric}
\usetikzlibrary{calc}
\definecolor{tensorblue}{rgb}{0.8,0.9,1}
\tikzset{ten/.style={fill=tensorblue}}

\usepackage{authblk} 
\usepackage[numbers,sort&compress]{natbib}
\usepackage{tikzit}
\usepackage{tikzscale}

\begin{document}

\title{Distilling Qubit Unitary Operations: A No-Go Theorem and Minimal Realization}

\author[1]{Jiayi Zhao}
\author[1]{Yu-Ao Chen}
\author[1]{Guocheng Zhen}
\author[1]{Chengkai Zhu}
\author[2]{Ranyiliu Chen\thanks{Corresponding author: Ranyiliu Chen. email: chenranyiliu@quantumsc.cn.}}
\author[1]{Xin Wang\thanks{Corresponding author: Xin Wang. email: felixxinwang@hkust-gz.edu.cn.}}
\affil[1]{\small Thrust of Artificial Intelligence, Information Hub,\par The Hong Kong University of Science and Technology (Guangzhou), Guangdong 511453, China}
\affil[2]{Quantum Science Center of Guangdong-Hong Kong-Macao Greater Bay Area, Shenzhen 518045, China}

\markboth{}%
{Shell \MakeLowercase{\textit{et al.}}: A Sample Article Using IEEEtran.cls for IEEE Journals}

\IEEEpubid{}

\maketitle

\begin{abstract}
Quantum gates executed on physical hardware are inevitably degraded by environmental noise. While state purification effectively distills static quantum resources, the dynamic execution of quantum algorithms requires a higher-order approach to mitigate errors on the operations themselves. In this work, we investigate \emph{universal unitary purification}: the task of utilizing a quantum higher-order operation to partially restore the ideal action of an unknown unitary corrupted by a known noise model. We consider deterministic purification of unknown single-qubit unitaries under independent depolarizing noise with ideal auxiliary operations. We prove that within the indefinite causal order (ICO) framework, no nontrivial 2-slot higher-order operation can universally purify the set of single-qubit unitaries. Overcoming this strict limitation, we establish that a 3-slot parallel architecture provides the minimal realization for nontrivial purification. We analytically derive the optimal average fidelity over all 3-slot ICO strategies and show that a parallel strategy attains it, strictly surpassing trivial strategies through input encoding and joint decoding of the channel outputs and a quantum memory. Furthermore, we provide a concrete quantum circuit construction attaining this optimum. Our results clarify the possibilities and limitations of purifying unknown qubit unitaries under depolarizing noise, offering insights for robust gate design.
\end{abstract}


\section{Introduction}
\IEEEPARstart{T}{he} pursuit of quantum technologies with superior performance is central to the advancement of quantum information processing, promising massive speedups for complex computational problems \cite{shor1997,biamonte2017quantum,Cao2019}. However, the intrinsic susceptibility of quantum devices to environmental disturbances significantly hinders their operational integrity. In practice, quantum systems inevitably interact with their environment, leading to deviations from the ideal scenario of pure states and unitary dynamics \cite{Nielsen2000,Preskill2018}.

To address this challenge, several approaches have been developed, including quantum error correction \cite{steane1996error} and error mitigation \cite{cai2023quantum}. Alongside these techniques, quantum state purification \cite{cirac1999optimal,keyl2001rate,fiuravsek2004optimal,childs2025streaming,yao2025protocols,he2026no,zhao2026power} has emerged as a vital strategy. When multiple copies of a noisy quantum state are available, state purification protocols can be applied to  `distill' a state with a higher fidelity to the initial ideal pure state. State-of-the-art quantum state purification has evolved from foundational global operations, such as the Cirac-Ekert-Macchiavello (CEM) protocol \cite{cirac1999optimal}, to highly specialized and resource-efficient methods. Recent advancements include streaming purification \cite{childs2025streaming} and generalized semidefinite programming (SDP) frameworks that map fidelity-probability trade-offs while minimizing ancilla overhead via parameterized quantum circuits \cite{yao2025protocols}. Li et al.\ establish optimal purity amplification for depolarized states in arbitrary dimensions and characterize its asymptotic fidelity~\cite{li2024optimal}; Scharnhorst et al.\ give nonasymptotic bounds on the number of copies sufficient to reach a target fidelity~\cite{scharnhorst2026nonasymptotic}. Moving into the operationally constrained regimes, recent findings establish no-go theorems under classically simulatable \cite{he2026no} or local (LOCC) \cite{zhao2026power} conditions for certain symmetric state sets.

\begin{figure}[t]
    \centering
    \includegraphics[width=1\linewidth]{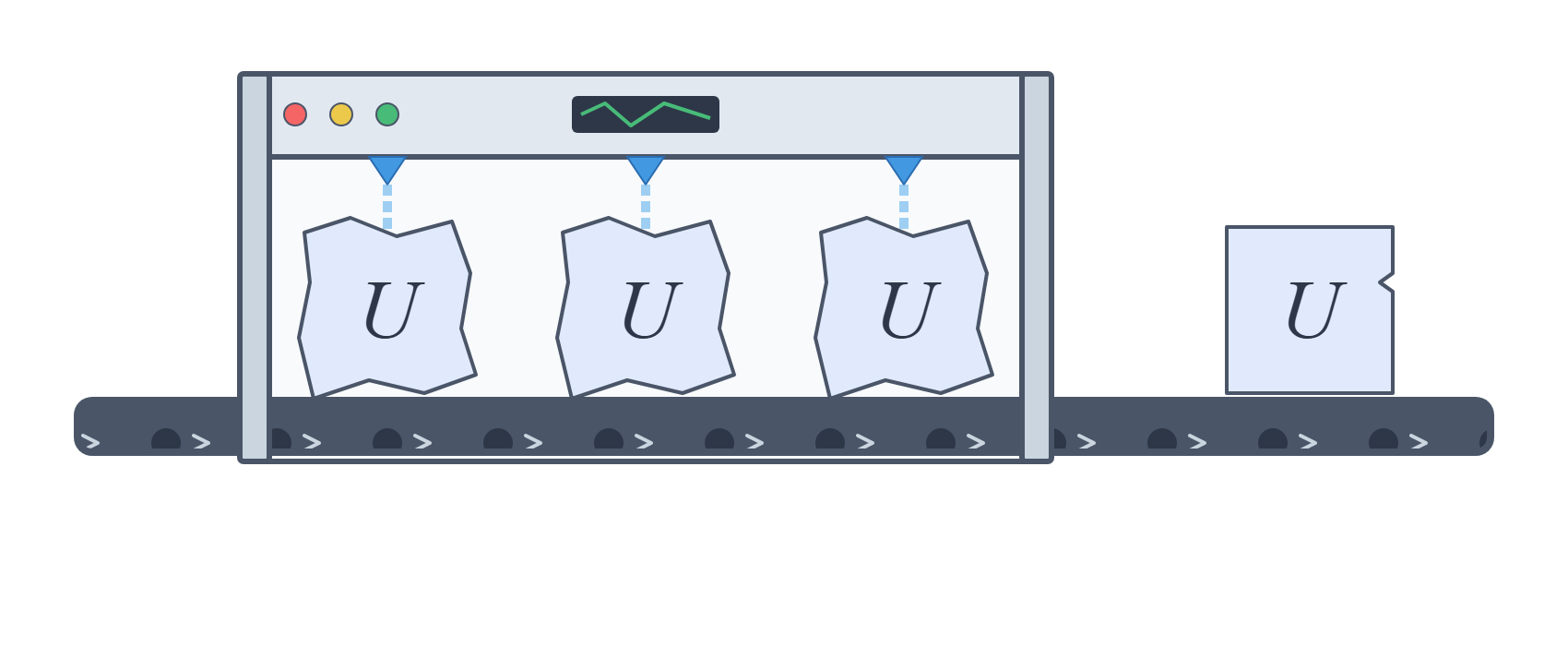}
    \caption{Conceptual illustration of unitary purification. Multiple copies of a noisy unitary channel (depicted as the jagged $U$ blocks on the left) are processed through a purification mechanism to produce a single output channel approximating the ideal $U$.} 
    \label{fig:up}
\end{figure}

At the channel level, virtual channel purification reconstructs expectation values associated with a purified channel and also admits a postselected physical variant~\cite{liu2025virtual}; both forms have been demonstrated experimentally in a photonic platform~\cite{fei2026experimental}. Symmetric channel verification exploits known symmetries of the target channel and includes virtual and feedback-based correction variants~\cite{tsubouchi2025symmetric}.

Channel access allows the input to be encoded before the noisy channels act, whereas purification of supplied state copies starts from states that are already noisy. Since states have no input ports and cannot receive information from earlier operations, sequential strategies offer no advantage over parallel ones when arbitrary joint processing of independent copies is allowed~\cite{Chiribella2009PRA}. Yet, when manipulating operations rather than static states, temporal ordering may become a nontrivial resource, raising the fundamental question of whether a performance gap exists between parallel and sequential architectures. Together, these considerations point toward a fundamental question: If we are given multiple access to an unknown, polluted quantum channel, can we distill an operation that more closely approximates the ideal, clean unitary? 

Motivated by this question, we study deterministic universal purification of unknown noisy unitary channels (see Fig.~\ref{fig:up}). Formally, we model this scenario using the framework of quantum higher-order operations (higher-order maps that transform quantum channels into quantum channels) in the language of quantum strategies \cite{Chiribella2008PRL,Chiribella2009PRA}. Given $n$ uses of an unknown channel representing a polluted version of an ideal unitary channel $\mathcal{U}$ under noise $\mathcal{N}$, we seek to find a universal higher-order operation (or quantum strategy) $\mathcal{C}$ that outputs a purified channel. The requirement of \emph{universal purification} is that $\mathcal{C}$ preserves or improves the channel fidelity of $\mathcal{N}\circ\mathcal{U}$ to the ideal unitary for every possible target $\mathcal{U}$.

Applying this framework, we investigate the feasibility of universal unitary purification under depolarizing noise, a canonical quantum error model. We first reveal a fundamental obstruction: no nontrivial 2-slot strategy can universally purify the set of single-qubit unitaries affected by depolarizing noise, even for indefinite causal order (ICO) ones. Specifically, we show that the optimal strategy trivially reduces to completely discarding one input and applying the identity operation to the other. Going beyond this limitation, we derive a closed-form upper bound on the average fidelity of all 3-slot ICO strategies and construct an explicit parallel circuit that attains it. This optimum surpasses the fidelity of trivial strategies. Overall, our results establish both the strict boundaries and the operational capabilities of distilling clean operations from noisy quantum gates in the qubit case with the canonical depolarizing noise model.

\section{Notation and Definitions}

\subsection{Quantum strategy framework}

Let $d$ denote the dimension of the Hilbert space on which the unknown ideal unitary $U$ acts. We represent quantum operations using the Choi-Jamio{\l}kowski isomorphism~\cite{Choi1975}.
The unnormalised maximally entangled state is denoted as $\dket{I}=\sum_i|ii\rangle$,
such that the Choi vector of a unitary $U$ is given by $\dket{U}=(U\ox I)\dket{I}$. For a general quantum channel $\mathcal{A}$, its corresponding Choi operator is denoted by $J^{\mathcal{A}}$.

Quantum higher-order operations that map quantum channels to quantum channels are described as quantum strategies~\cite{Chiribella2008PRL,Chiribella2009PRA}. An $n$-slot quantum strategy $\cC$ is mathematically represented by its Choi-Jamio{\l}kowski operator $C_{P I^n O^n F}$, where $I^n O^n = I_1 O_1 I_2 O_2 \cdots I_n O_n$ denotes the internal slot systems, $P$ is the initial preparation (past) system, and $F$ is the final output (future) system. Link products~\cite{Chiribella2009PRA}, denoted by $*$, compose operations by contracting the tensor factors of operators sharing a common subsystem.

\begin{figure}[!t]
\centering
\scalebox{0.55}{\begin{tikzpicture}[scale=1.00]
	\draw (0,5.75) -- (0,1) -- (22.25,1) -- (22.25,5.75) 
		-- (20.5,5.75) -- (20.5,3) 
		-- (17,3) -- (17,5.75) 
		-- (15.25,5.75) -- (15.25,3) 
		-- (12.25,3) -- (12.25,5.75) 
		-- (10.5,5.75) -- (10.5,3) 
		-- (7,3) -- (7,5.75) 
		-- (5.25,5.75) -- (5.25,3) 
		-- (1.75,3) -- (1.75,5.75) -- cycle;

	\node at (11.125, 2) {$\mathcal{C}$};

	\draw (-1,4.75) -- (0,4.75);
	\node at (-0.5, 5.25) {$P$};

	\draw (1.75,4.75) -- (2.75,4.75);
	\node at (2.25, 5.25) {$I_1$};
	\draw (4.25,4.75) -- (5.25,4.75);
	\node at (4.75, 5.25) {$O_1$};

	\draw (7,4.75) -- (8,4.75);
	\node at (7.5, 5.25) {$I_2$};
	\draw (9.5,4.75) -- (10.5,4.75);
	\node at (10, 5.25) {$O_2$};

	\draw (12.25,4.75) -- (12.75,4.75);
	\node at (13.75, 4.75) {$\cdots$};
	\draw (14.75,4.75) -- (15.25,4.75);

	\draw (17,4.75) -- (18,4.75);
	\node at (17.5, 5.25) {$I_n$};
	\draw (19.5,4.75) -- (20.5,4.75);
	\node at (20, 5.25) {$O_n$};

	\draw (22.25,4.75) -- (23.25,4.75);
	\node at (22.75, 5.25) {$F$};
\end{tikzpicture}}
\caption{Illustration of the $n$-slot sequential strategy structure $\cC$, where operations are applied in a strictly fixed causal order.}
\label{fig:seq1}
\end{figure}

\begin{figure}[!t]
\centering
\scalebox{0.7}{\begin{tikzpicture}[scale=0.8]
	\begin{pgfonlayer}{nodelayer}
		\node [style=none] (0) at (3, 6) {$P$};
		\node [style=none] (1) at (17.75, 6) {$F$};
		\node [style=none] (2) at (10, 11.25) {$\mathcal{C}$};
		\node [style=none] (5) at (5, 16.88) {};
		\node [style=none] (6) at (7, 16.88) {};
		\node [style=none] (7) at (7, 12.38) {};
		\node [style=none] (8) at (13.5, 12.38) {};
		\node [style=none] (9) at (13.5, 16.88) {};
		\node [style=none] (10) at (15.5, 16.88) {};
		\node [style=none] (11) at (15.5, -4.88) {};
		\node [style=none] (12) at (13.5, -4.88) {};
		\node [style=none] (13) at (13.5, -0.38) {};
		\node [style=none] (14) at (7, -0.38) {};
		\node [style=none] (15) at (7, -4.88) {};
		\node [style=none] (16) at (5, -4.88) {};
		\node [style=none] (25) at (4, 6) {};
		\node [style=none] (26) at (5, 6) {};
		\node [style=none] (27) at (15.5, 6) {};
		\node [style=none] (28) at (16.75, 6) {};
		\node [style=none] (29) at (7, 14.63) {};
		\node [style=none] (32) at (13.5, 14.63) {};
		\node [style=none] (37) at (7.5, 15.15) {$I_1$};
		\node [style=none] (38) at (13, 15.15) {$O_1$};
		\node [style=none] (41) at (7.75, 14.63) {};
		\node [style=none] (42) at (12.75, 14.63) {};
		\node [style=none] (33) at (7, -2.63) {};
		\node [style=none] (36) at (13.5, -2.63) {};
		\node [style=none] (39) at (7.5, -2.1) {$I_n$};
		\node [style=none] (40) at (13, -2.1) {$O_n$};
		\node [style=none] (43) at (7.75, -2.63) {};
		\node [style=none] (44) at (12.75, -2.63) {};
		\node [style=none] (45) at (7, 9.75) {};
		\node [style=none] (46) at (7, 6.75) {};
		\node [style=none] (47) at (13.5, 9.75) {};
		\node [style=none] (48) at (13.5, 6.75) {};
		\node [style=none] (53) at (7, 8.25) {};
		\node [style=none] (54) at (13.5, 8.25) {};
		\node [style=none] (55) at (7.5, 8.78) {$I_2$};
		\node [style=none] (56) at (13, 8.78) {$O_2$};
		\node [style=none] (57) at (7.75, 8.25) {};
		\node [style=none] (58) at (12.75, 8.25) {};
		\node [style=none] (49) at (7, 4.88) {};
		\node [style=none] (50) at (7, 1.88) {};
		\node [style=none] (51) at (13.5, 4.88) {};
		\node [style=none] (52) at (13.5, 1.88) {};
		\node [style=none] (59) at (7, 3.38) {};
		\node [style=none] (60) at (13.5, 3.38) {};
		\node [style=none] (61) at (7.75, 3.38) {};
		\node [style=none] (62) at (12.75, 3.38) {};
		\node [style=none] (63) at (10, 3.38) {$\vdots$};
	\end{pgfonlayer}
	\begin{pgfonlayer}{edgelayer}
		\draw [style=none] (16.center) to (5.center);
		\draw [style=none] (5.center) to (6.center);
		\draw [style=none] (6.center) to (7.center);
		\draw [style=none] (7.center) to (8.center);
		\draw [style=none] (8.center) to (9.center);
		\draw [style=none] (9.center) to (10.center);
		\draw [style=none] (10.center) to (11.center);
		\draw [style=none] (11.center) to (12.center);
		\draw [style=none] (12.center) to (13.center);
		\draw [style=none] (13.center) to (14.center);
		\draw [style=none] (14.center) to (15.center);
		\draw [style=none] (15.center) to (16.center);
		\draw [style=none] (25.center) to (26.center);
		\draw [style=none] (27.center) to (28.center);
		\draw (29.center) to (41.center);
		\draw (32.center) to (42.center);
		\draw (33.center) to (43.center);
		\draw (36.center) to (44.center);
		\draw (45.center) to (47.center);
		\draw (48.center) to (47.center);
		\draw (45.center) to (46.center);
		\draw (46.center) to (48.center);
		\draw (49.center) to (51.center);
		\draw (52.center) to (51.center);
		\draw (49.center) to (50.center);
		\draw (50.center) to (52.center);
		\draw (53.center) to (57.center);
		\draw (54.center) to (58.center);
		\draw (59.center) to (61.center);
		\draw (60.center) to (62.center);
	\end{pgfonlayer}
\end{tikzpicture}}
\caption{Illustration of the $n$-slot indefinite causal order (ICO) strategy structure $\cC$, with no prescribed global order of the slots.}
\label{fig:ico1}
\end{figure}

The fundamental distinction among different routing strategies manifests entirely in the linear constraints imposed on their Choi-Jamio{\l}kowski operators. To rigorously generalise these constraints to $n$ slots, we formalise the notation ${}_{X} C := \frac{I_X}{d_X} \ox \tr_X[C]$ to denote the operation of taking the partial trace over a subsystem $X$ and replacing it with the normalised identity matrix, where $d_X$ is the dimension of system $X$. In a parallel strategy, neither the slot outputs $O_i$ nor the final output $F$ can signal to any slot input $I_j$.

In a sequential strategy, the operations follow a strict, predetermined causal order, meaning the action of the first operation causally precedes the second, progressing linearly up to the $n$-th operation, illustrated in Fig.~\ref{fig:seq1}. The principle of causality dictates that the operation at slot $k$ cannot depend on the outputs of any future slots $k'>k$. In the Choi representation, a sequential strategy must satisfy the following linear conditions:
\begin{equation}
\begin{cases}
    C \ge 0 \\
    {}_{F} C = {}_{O_n F} C \\
    {}_{I_k O_k \cdots I_n O_n F} C = {}_{O_{k-1} I_k O_k \cdots I_n O_n F} C \quad \forall k \in \{2, \dots, n\} \\
    {}_{I^n O^n F} C = {}_{P I^n O^n F} C \\
    \tr[C] = d_P \prod_{i=1}^n d_{I_i}
\end{cases}
\end{equation}

In contrast, the ICO class does not impose a fixed global order of the slots (Fig.~\ref{fig:ico1}). We consider deterministic completely positive supermaps: filling all $n$ slots with arbitrary quantum channels $I_i\to O_i$ must yield a quantum channel from $P$ to $F$. Here $d_{I_i}=d_{O_i}=d$. Let $S\subseteq\{1,\dots,n\}$ be nonempty and $\bar{S}$ its complement. In our Choi convention, the standard characterization~\cite{araujo2017purification,milz2024characterising} reads, for every such $S$:
\begin{equation}\label{eq:ico1}
\begin{cases}
    C \ge 0 \\
    \sum_{S' \subseteq S} (-1)^{|S'|} {}_{I_{\bar{S}} O_{\bar{S}} O_{S'} F} C = 0 \\
    {}_{I^n O^n F} C = {}_{P I^n O^n F} C \\
    \tr[C] = d_P \prod_{i=1}^n d_{I_i}
\end{cases}
\end{equation}
The positivity condition $C\ge0$ ensures complete positivity of the supermap. The remaining equations are equivalent to requiring a trace-preserving output for every choice of inserted channels. Since the output depends linearly on each slot channel, its normalization can depend on individual slots as well as combinations of slots. The subset equations remove this dependence, while the last two equations fix the normalization.

For $n=1$, the subset equation reduces to ${}_{F}C={}_{O_1F}C$. For $n=2$, setting $S=\{1,2\}$ gives ${}_{F} C + {}_{O_1 O_2 F} C = {}_{O_1 F} C + {}_{O_2 F} C$, while $S=\{1\}$ and $S=\{2\}$ give ${}_{I_2 O_2 F} C = {}_{O_1 I_2 O_2 F} C$ and ${}_{I_1 O_1 F} C = {}_{I_1 O_1 O_2 F} C$, respectively.

\subsection{Problem formulation}

The noise corrupting a quantum gate is modelled by a completely positive and trace-preserving (CPTP) quantum channel $\cN$. Given an unknown ideal (noiseless) unitary channel $\cU(\cdot)=U(\cdot)U^\dagger$, the corresponding physical implementation is the noisy gate $\cN\circ\cU$. The fundamental task of unitary purification is to recover a higher-fidelity approximation of the ideal unitary $\cU$ given $n$ uses of the noisy gate $\cN\circ\cU$, potentially adaptively.

Suppose we are given a noisy operation $\cN\circ\cU$, represented by its Choi operator $J^{\cN\circ\cU} \coloneqq J^{\cN} * \dketbra{U}{U}$ where $U$ is the target pure unitary and $\cN$ is a given noise channel. Inserting $n$ copies of this noisy operation into the slots of the comb $\mathcal{C}$ yields an effective channel from $P$ to $F$, denoted by $\mathcal{E}_U$. The corresponding Choi operator is given by
\begin{equation}
    J^{\mathcal{E}_U} \coloneqq C_{PI^nO^nF} * \left(J^{\mathcal{N}\circ\mathcal{U}}\right)_{I^nO^n}^{\otimes n}. 
    \label{eq:output_choi}
\end{equation}
Given the Choi operator in~\eqref{eq:output_choi}, the channel fidelity with respect to the target unitary $\cU$ is evaluated as $F(J^{\mathcal{E}_U}, \dketbra{U}{U}) := \tr[J^{\mathcal{E}_U}\,\dketbra{U}{U}]/d^2$. We take the fidelity of the bare noisy channel $\cN\circ\cU$ to the same ideal target $\cU$ as the baseline. If the output fidelity meets or exceeds this baseline for every $U\in SU(d)$, we call $\cC$ a \emph{universal} purification protocol for $\cN$.

\begin{definition}[Universal unitary purification protocol]
\label{def:purification_protocol}
A $n$-slot quantum strategy $\cC$ is a \emph{universal unitary purification protocol} for a noisy channel $\cN$ if it increases the channel fidelity for every $U\in SU(d)$, i.e.,
\begin{equation}
    F\!\left(J^{\mathcal{E}_U}, \dketbra{U}{U}\right)
    \;\geq\;
    F\!\left(J^{\cN\circ\cU},\, \dketbra{U}{U}\right),
    \label{eq:purification_condition}
\end{equation}
for all $U \in SU(d)$.
\end{definition}

We remark that the \emph{trivial} purification protocol, which
preserves the fidelity for every $U$, also fits the above definition (equality holds
in~\eqref{eq:purification_condition} for all $U\in SU(d)$.) On the contrary, a \emph{nontrivial} purification protocol is defined as a protocol that strictly increases the fidelity for at least one unitary. 

To quantify and rigorously compare the performance of different purification higher-order operations, we define the \emph{average fidelity} over all possible target unitaries as
\begin{equation}
    F_{\mathrm{ave}}(\cN, \cC)
    \coloneqq \int \mathrm{d}U\;
       \tr\!\left[
           J^{\mathcal{E}_U} \,\dketbra{U}{U}/d^2
       \right],
    \label{eq:avg_fidelity}
\end{equation}
where the integral is over the Haar measure $\mathrm{d}U$ on $SU(d)$. Let $\mathcal{S}_n^{\mathsf{R}}$ denote the set of $n$-slot strategies in class $\mathsf{R}\in\{\mathrm{Par},\mathrm{Seq},\mathrm{ICO}\}$, corresponding to parallel, sequential, and indefinite-causal-order strategies, respectively. We write
\begin{equation}
    F_{\max}^{\mathsf{R}}(n;\cN)
    \coloneqq
    \max_{\cC\in\mathcal{S}_n^{\mathsf{R}}}
    F_{\mathrm{ave}}(\cN,\cC)
    \label{eq:max_avg_fidelity}
\end{equation}
for the corresponding optimal average fidelity. It is then evident from the definition that the average fidelity provides a witness of \emph{non-universality}: if a purification protocol yields zero improvement in the average fidelity defined in~\eqref{eq:avg_fidelity}, then the protocol fails to be non-trivially universal.

Finally, in this work we will primarily focus on the qubit depolarizing channel $\cN_{\gamma}$, defined as $\mathcal{N}_\gamma(\rho)=(1-\gamma)\rho+\gamma I/d$, where $\gamma\in(0,1)$ parametrizes the noise level and $d=2$ is the Hilbert-space dimension. We consider deterministic strategies without postselection. The noise acts independently on each channel use, with the same strength $\gamma$. Auxiliary state preparation, quantum memory, and processing operations are assumed ideal.

\section{Fundamental Limits of 2-Slot Architectures}
\label{sec:nogo}

Our first main result reveals that no nontrivial universal $2$-slot purification strategy exists for qubit depolarizing noise:


\begin{theorem}[No-go for $2$-slot purification]
\label{thm:nogo}
For all $\gamma\in(0,1)$, the depolarizing channel $\cN_{\gamma}$ admits no nontrivial $2$-slot ICO
strategy satisfying the universal purification
condition~\eqref{eq:purification_condition}. In fact,
\begin{equation}
    F_{\max}^{\mathrm{ICO}}(2;\cN_\gamma)=1-\frac{3}{4}\gamma.
\end{equation}
\end{theorem}
\noindent\underline{\emph{Proof sketch.}} We prove that no $2$-slot ICO strategy can achieve any nontrivial fidelity improvement by showing that the optimal average fidelity improvement over all such protocols is bounded by zero. 

To quantify the average fidelity advantage of a general strategy $\mathcal{C}$ over a trivial purification protocol, we evaluate the objective via a trace expression $\operatorname{Tr}[C\Omega]$. In this formulation, the performance operator $\Omega$ absorbs all components independent of the optimization variable $C$ representing the strategy $\mathcal{C}$ (see Appendix~\ref{app:proof_universal_no-go} for the explicit expression).  

Suppose, by contradiction, that there exists a $2$-slot ICO strategy $\check{\mathcal{C}}$ strictly exceeding the trivial bound. For $V,W\in\mathrm{SU}(2)$, define
\begin{equation}
    U_{V,W}=V_P\otimes V^{*\otimes2}_{I^2}\otimes W^{*\otimes2}_{O^2}\otimes W_F.
\end{equation}
The performance operator $\Omega$ satisfies $U_{V,W}^\dagger\Omega U_{V,W}=\Omega$ for every such pair. By cyclicity of the trace, conjugating $\check{C}$ by $U_{V,W}$ therefore leaves the objective unchanged:
\begin{equation*}
\begin{aligned}
    \operatorname{Tr}\!\left[\Omega U_{V,W}\check{C}U_{V,W}^\dagger\right]
    &=\operatorname{Tr}\!\left[U_{V,W}^\dagger\Omega U_{V,W}\check{C}\right]\\
    &=\operatorname{Tr}(\Omega\check{C}).
\end{aligned}
\end{equation*}
Thus every transformed strategy has the same objective value as $\check{C}$. These local unitary conjugations also preserve the ICO conditions. We then twirl the Choi operator over the Haar measure:
\begin{equation} \label{eq:twirl1}
    \bar{C}=\mathbb E_{V,W}\bigl[U_{V,W}\check{C}U_{V,W}^\dagger\bigr].
\end{equation}
Averaging preserves positive semidefiniteness and the linear ICO normalization conditions. By linearity of the trace, its objective value is the average of the unchanged objective values, so $\operatorname{Tr}(\Omega\bar{C})=\operatorname{Tr}(\Omega\check{C})$. Therefore, $\bar{C}$ remains a valid ICO strategy that strictly exceeds the trivial bound. Moreover, it explicitly commutes with the group action, satisfying $[\bar{C},U_{V,W}]=0$ for all $V,W\in\mathrm{SU}(2)$.

Next, we apply the Clebsch-Gordan decomposition to the representation of $\mathrm{SU}(2) \times \mathrm{SU}(2)$. Invoking Schur's lemma, the commutation relation forces $\bar{C}$ to be block-diagonal in the irreducible decomposition basis. Specifically, there exists a unitary transformation $\tilde{G}$ such that $\bar{C}$ reduces to a handful of scalar parameters labeling the blocks:
\begin{equation} \label{eq:block_diag}
    \tilde{G} \bar{C} \tilde{G}^\dagger = \operatorname{diag}(H_0,\, H_1\otimes I_2,\, H_2\otimes I_2,\, h_3 I_4) \otimes I_4,
\end{equation}
where $H_0 \in \mathbb{C}^{4\times 4}$, $H_1, H_2 \in \mathbb{C}^{2\times 2}$, and $h_3 \in \mathbb{C}$ are all positive semidefinite matrices or scalars. 

This structural constraint drastically reduces the intractable optimization over the full ICO set into a low-dimensional SDP. The ICO linear feasibility constraints trace out to a set of elementary relations among the blocks. By analyzing the feasible region of this SDP (see Appendix for full constraint substitutions), we identify that the objective function reaches its boundary maximum only when $H_{200}=0$, $H_{000}=3H_{033}$, and $H_{022}=3 H_{011}$. 

The optimization is ultimately simplified to maximizing a quadratic form of a $2\times 2$ matrix $M$:
\begin{equation} \label{eq:reduced_sdp}
    \max \quad \vec{v}^T M \vec{v} - \frac{5}{2}(1-\gamma),
\end{equation}
subject to the trace constraint $\|\vec{v}\|^2 \le \frac{1}{2}$, where we define the vector $\vec{v} = (\sqrt{H_{011}}, \sqrt{H_{033}})^T$ and the matrix $M$ is given by
\begin{equation} \label{eq:matrix_M}
    M = (1-\gamma) \begin{pmatrix}
        -1 & 2\sqrt{3} \\
        2\sqrt{3} & 3
    \end{pmatrix}.
\end{equation}
The maximal eigenvalue of $M$ can be analytically computed as $5(1-\gamma)$, with the corresponding normalized eigenstate $\frac{1}{2}(1, \sqrt{3})^T$. Substituting this maximal eigenvalue and the upper bound $\|\vec{v}\|^2 = \frac{1}{2}$ into Eq.~\eqref{eq:reduced_sdp}, the maximum achievable value is exactly:
\begin{equation}
    5(1-\gamma) \times \frac{1}{2} - \frac{5}{2}(1-\gamma) = 0.
\end{equation}
This indicates that the initially assumed strictly positive fidelity improvement is impossible and is bounded above by zero. This contradicts the assumption that $\check{\mathcal{C}}$ exceeds the trivial bound, completing the proof. \hfill$\square$

The detailed proof of Theorem~\ref{thm:nogo} is provided in Appendix~\ref{app:proof_universal_no-go}.


\begin{remark}
    We note that this fundamental no-go result extends beyond the full unitary group to subsets exhibiting sufficient symmetry. For example, Theorem~\ref{thm:nogo} applies to any unitary subset forming a unitary 3-design (\emph{e.g.,} the Clifford group), as the SDP performance evaluation involves terms up to the third tensor power of $U$. However, this structural obstruction vanishes when restricting target operations to less symmetric gate sets. For the specific gate set $\{X,Y,Z,I,H,S\}$, we conduct numerical experiments to compute the optimal average fidelity after purification across different strategies. The SDPs are solved by MATLAB~\cite{MATLAB} with the interpreters CVX~\cite{cvx,gb08} and QETLAB~\cite{qetlab}. We observe distinct performance gaps between parallel, sequential, and indefinite causal order strategies under depolarizing noise (see Fig.~\ref{fig:causalorder}), indicating that the relaxation of symmetry constraints enables indefinite causal order strategies to substantially outperform conventional approaches, thereby indicating a potential regime where causal structures could yield advantages for mitigating error.
\end{remark}

\begin{figure}
    \centering
    \includegraphics[width=1\linewidth]{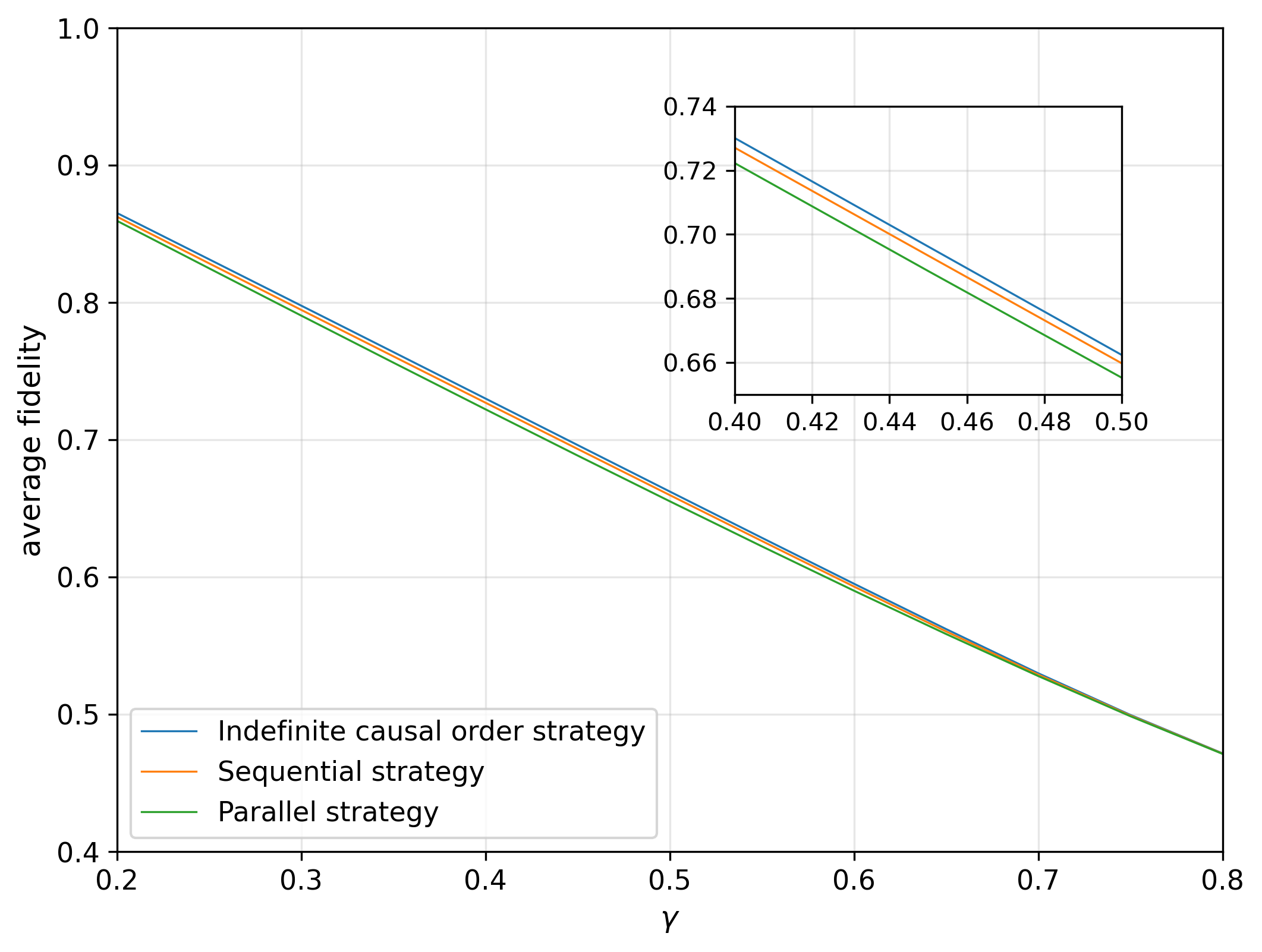}
    \caption{Optimal average fidelity for purification of $\{X,Y,Z,I,H,S\}$ against depolarizing noise under parallel, sequential, and indefinite causal order strategies.}
    \label{fig:causalorder}
\end{figure}


\section{Analytic Optimal Fidelity and Circuit Realization of the 3-Slot Architecture}

Since two slots cannot yield nontrivial universal purification, at least three channel uses are required. We now determine the optimal three-slot ICO fidelity and construct a parallel strategy attaining it.

\begin{theorem}[Optimal $3$-slot purification]
\label{thm:go}
For all $\gamma\in(0,1)$, the optimal average fidelity over all $3$-slot ICO strategies under the depolarizing channel $\cN_\gamma$ is
\begin{equation}\label{eq:fidelity}
\begin{aligned}
    F_{\max}^{\mathrm{ICO}}(3;\cN_\gamma)
    &=\frac{2-\gamma}{24}\bigl[12+(8\sqrt2-9)\gamma\\
    &\hspace{29mm}+(3-8\sqrt2)\gamma^2\bigr].
\end{aligned}
\end{equation}
This optimum is attained by the explicit $3$-slot parallel strategy constructed in Appendix~\ref{sec:cir}, which satisfies the universal purification condition~\eqref{eq:purification_condition} with strict inequality for every $U\in\mathrm{SU}(2)$. Consequently, the optimal average fidelities coincide: $F_{\max}^{\mathrm{Par}}(3;\cN_\gamma)=F_{\max}^{\mathrm{Seq}}(3;\cN_\gamma)=F_{\max}^{\mathrm{ICO}}(3;\cN_\gamma)$.

\end{theorem} 

\noindent\underline{\emph{Proof sketch.}}
By the same twirling argument as in the first proof sketch, Eq.~\eqref{eq:twirl1} gives a valid three-slot ICO strategy $\bar{C}$ with unchanged average fidelity and $[\bar{C},U_{V,W}]=0$ for all $V,W\in\mathrm{SU}(2)$, where
\begin{equation}
    U_{V,W}=V_P\otimes V^{*\otimes3}_{I^3}\otimes W^{*\otimes3}_{O^3}\otimes W_F.
\end{equation}

We further observe that the performance operator is invariant under any permutation $\pi\in S_3$ of the three identical slots $(I_k,O_k)$. The ICO conditions in Eq.~\eqref{eq:ico1} are imposed for every nonempty subset of slots, so relabelling the slots only permutes these conditions. Averaging over the six slot permutations therefore again preserves validity and fidelity. Denoting the unitary representation of $\pi$ by $P_\pi$, we may thus impose $\bar{C}=P_\pi\bar{C}P_\pi^\dagger$ without loss of optimality.

Next, we apply the Clebsch--Gordan decomposition to the representation $(V,W)\mapsto V\otimes V^{*\otimes 3}\otimes W^{*\otimes 3}\otimes W$ of $\mathrm{SU}(2)\times\mathrm{SU}(2)$. By Schur's lemma, the commutation relation $[\bar{C},U_{V,W}]=0$ forces $\bar{C}$ to be block-diagonal in the corresponding irreducible decomposition basis. Specifically, there exists a $256\times256$ unitary $G$, where $256=2^8$ is the total dimension of the eight qubit ports, such that
\begin{equation}
    G^\dagger\bar{C}G=\bigoplus_{k=0}^{8}(I_{d_k}\otimes H_k),
\end{equation}
where $H_k$ are positive semidefinite matrices on the multiplicity spaces and $d_k$ are the irreducible representation dimensions.

This structure, combined with the $S_3$ symmetry and the ICO normalization conditions, reduces the optimization to an SDP over the blocks $\{H_k\}$. Using these linear constraints and the positive semidefiniteness of the blocks, we establish that the average fidelity cannot exceed the value in Eq.~\eqref{eq:fidelity}. The analytic certificate and the detailed upper-bound argument are given in Appendix~\ref{app:proof_universal_go}.

For achievability, we construct an explicit parallel 3-slot strategy. An encoder isometry $V_{\mathrm{enc}}$ maps the input system $P$ to the three channel inputs $I_1,I_2,I_3$ and a one-qubit memory register $R$. The three uses of $\cN_\gamma\circ\cU$ are then applied in parallel, producing outputs $O_1,O_2,O_3$. A decoder isometry $V_{\mathrm{dec}}$ processes $R,O_1,O_2,O_3$, outputs a single system qubit $F$, and discards four ancillary output qubits\footnote{The explicit isometries of encoder and decoder can be found at our GitHub repository \url{https://github.com/jyzhau/Unitary-Purification}.}. This construction, illustrated in Fig.~\ref{fig:par}, attains the value in Eq.~\eqref{eq:fidelity}. We also provide the circuit implementation in Appendix~\ref{sec:cir}. Since parallel strategies are sequential strategies and sequential strategies are ICO strategies, this construction attains the optimum in all three classes.
\hfill$\square$

\begin{figure}[t]
    \centering
    \scalebox{0.75}{\begin{tikzpicture}[scale=1.452]
	\begin{pgfonlayer}{nodelayer}
		\node [style=none] (220) at (48.5, 8.75) {};
		\node [style=none] (221) at (50.25, 8.75) {};
		\node [style=none] (222) at (50.25, 1.75) {};
		\node [style=none] (223) at (48.5, 1.75) {};
		\node [style=none] (224) at (58.5, 8.75) {};
		\node [style=none] (225) at (60.25, 8.75) {};
		\node [style=none] (226) at (60.25, 1.75) {};
		\node [style=none] (227) at (58.5, 1.75) {};
		\node [style=none] (228) at (54.75, 8.75) {};
		\node [style=none] (229) at (56.5, 8.75) {};
		\node [style=none] (230) at (56.5, 7.75) {};
		\node [style=none] (231) at (54.75, 7.75) {};
		\node [style=none] (232) at (54.75, 6.75) {};
		\node [style=none] (233) at (56.5, 6.75) {};
		\node [style=none] (234) at (56.5, 5.75) {};
		\node [style=none] (235) at (54.75, 5.75) {};
		\node [style=none] (236) at (54.75, 4.75) {};
		\node [style=none] (237) at (56.5, 4.75) {};
		\node [style=none] (238) at (56.5, 3.75) {};
		\node [style=none] (239) at (54.75, 3.75) {};
		\node [style=none] (240) at (52, 8.75) {};
		\node [style=none] (241) at (53.5, 8.75) {};
		\node [style=none] (242) at (53.5, 7.75) {};
		\node [style=none] (243) at (52, 7.75) {};
		\node [style=none] (244) at (52, 6.75) {};
		\node [style=none] (245) at (53.5, 6.75) {};
		\node [style=none] (246) at (53.5, 5.75) {};
		\node [style=none] (247) at (52, 5.75) {};
		\node [style=none] (248) at (52, 4.75) {};
		\node [style=none] (249) at (53.5, 4.75) {};
		\node [style=none] (250) at (53.5, 3.75) {};
		\node [style=none] (251) at (52, 3.75) {};
		\node [style=none] (252) at (52.75, 8.25) {$\mathcal{U}$};
		\node [style=none] (253) at (52.75, 6.25) {$\mathcal{U}$};
		\node [style=none] (254) at (52.75, 4.25) {$\mathcal{U}$};
		\node [style=none] (255) at (48, 5.5) {};
		\node [style=none] (256) at (48.5, 5.5) {};
		\node [style=none] (257) at (50.25, 8.25) {};
		\node [style=none] (258) at (52, 8.25) {};
		\node [style=none] (259) at (50.25, 6.25) {};
		\node [style=none] (260) at (52, 6.25) {};
		\node [style=none] (261) at (50.25, 4.25) {};
		\node [style=none] (262) at (52, 4.25) {};
		\node [style=none] (263) at (56.5, 8.25) {};
		\node [style=none] (264) at (58.5, 8.25) {};
		\node [style=none] (265) at (56.5, 6.25) {};
		\node [style=none] (266) at (58.5, 6.25) {};
		\node [style=none] (267) at (56.5, 4.25) {};
		\node [style=none] (268) at (58.5, 4.25) {};
		\node [style=none] (269) at (50.25, 2.25) {};
		\node [style=none] (270) at (58.5, 2.25) {};
		\node [style=none] (271) at (60.25, 7.75) {};
		\node [style=none] (272) at (61, 7.75) {};
		\node [style=none] (273) at (60.25, 6.25) {};
		\node [style=none] (274) at (61, 6.25) {};
		\node [style=none] (275) at (60.25, 5) {};
		\node [style=none] (276) at (61, 5) {};
		\node [style=none] (277) at (60.25, 3.75) {};
		\node [style=none] (278) at (61, 3.75) {};
		\node [style=none] (279) at (60.25, 2.5) {};
		\node [style=none] (280) at (61, 2.5) {};
		\node [style=none] (281) at (61, 6.75) {};
		\node [style=none] (282) at (61, 5.75) {};
		\node [style=none] (283) at (61.75, 6.75) {};
		\node [style=none] (284) at (61.75, 5.75) {};
		\node [style=none] (285) at (62.25, 6.25) {};
		\node [style=none] (286) at (61, 5.5) {};
		\node [style=none] (287) at (61, 4.5) {};
		\node [style=none] (288) at (61.75, 5.5) {};
		\node [style=none] (289) at (61.75, 4.5) {};
		\node [style=none] (290) at (62.25, 5) {};
		\node [style=none] (291) at (61, 4.25) {};
		\node [style=none] (292) at (61, 3.25) {};
		\node [style=none] (293) at (61.75, 4.25) {};
		\node [style=none] (294) at (61.75, 3.25) {};
		\node [style=none] (295) at (62.25, 3.75) {};
		\node [style=none] (296) at (61, 3) {};
		\node [style=none] (297) at (61, 2) {};
		\node [style=none] (298) at (61.75, 3) {};
		\node [style=none] (299) at (61.75, 2) {};
		\node [style=none] (300) at (62.25, 2.5) {};
		\node [style=none] (301) at (49.375, 5.5) {$V_{\mathrm{enc}}$};
		\node [style=none] (302) at (59.375, 5.5) {$V_{\mathrm{dec}}$};
		\node [style=none] (303) at (55.625, 8.25) {$\mathcal{N}_\gamma$};
		\node [style=none] (304) at (55.625, 6.25) {$\mathcal{N}_\gamma$};
		\node [style=none] (305) at (55.625, 4.25) {$\mathcal{N}_\gamma$};
		\node [style=none] (306) at (48, 6.1) {$P$};
		\node [style=none] (307) at (51.625, 8.55) {$I_1$};
		\node [style=none] (308) at (51.625, 6.55) {$I_2$};
		\node [style=none] (309) at (51.625, 4.55) {$I_3$};
		\node [style=none] (310) at (57.875, 8.55) {$O_1$};
		\node [style=none] (311) at (57.875, 6.55) {$O_2$};
		\node [style=none] (312) at (57.875, 4.55) {$O_3$};
		\node [style=none] (313) at (53.375, 1.2) {$R$};
		\node [style=none] (314) at (60.625, 8.05) {$F$};
		\node [style=none] (315) at (61.45, 6.25) {\scriptsize trash};
		\node [style=none] (316) at (61.45, 5) {\scriptsize trash};
		\node [style=none] (317) at (61.45, 3.75) {\scriptsize trash};
		\node [style=none] (318) at (61.45, 2.5) {\scriptsize trash};
		\node [style=none] (319) at (53.5, 8.25) {};
		\node [style=none] (320) at (54.75, 8.25) {};
		\node [style=none] (321) at (53.5, 6.25) {};
		\node [style=none] (322) at (54.75, 6.25) {};
		\node [style=none] (323) at (53.5, 4.25) {};
		\node [style=none] (324) at (54.75, 4.25) {};
	\end{pgfonlayer}
	\begin{pgfonlayer}{edgelayer}
		\draw (220.center) to (221.center);
		\draw (221.center) to (222.center);
		\draw (222.center) to (223.center);
		\draw (223.center) to (220.center);
		\draw (224.center) to (225.center);
		\draw (225.center) to (226.center);
		\draw (226.center) to (227.center);
		\draw (227.center) to (224.center);
		\draw (228.center) to (229.center);
		\draw (229.center) to (230.center);
		\draw (230.center) to (231.center);
		\draw (231.center) to (228.center);
		\draw (232.center) to (233.center);
		\draw (233.center) to (234.center);
		\draw (234.center) to (235.center);
		\draw (235.center) to (232.center);
		\draw (236.center) to (237.center);
		\draw (237.center) to (238.center);
		\draw (238.center) to (239.center);
		\draw (239.center) to (236.center);
		\draw (240.center) to (241.center);
		\draw (241.center) to (242.center);
		\draw (242.center) to (243.center);
		\draw (243.center) to (240.center);
		\draw (244.center) to (245.center);
		\draw (245.center) to (246.center);
		\draw (246.center) to (247.center);
		\draw (247.center) to (244.center);
		\draw (248.center) to (249.center);
		\draw (249.center) to (250.center);
		\draw (250.center) to (251.center);
		\draw (251.center) to (248.center);
		\draw (281.center) to (283.center);
		\draw (283.center) to (285.center);
		\draw (285.center) to (284.center);
		\draw (284.center) to (282.center);
		\draw (282.center) to (281.center);
		\draw (286.center) to (288.center);
		\draw (288.center) to (290.center);
		\draw (290.center) to (289.center);
		\draw (289.center) to (287.center);
		\draw (287.center) to (286.center);
		\draw (291.center) to (293.center);
		\draw (293.center) to (295.center);
		\draw (295.center) to (294.center);
		\draw (294.center) to (292.center);
		\draw (292.center) to (291.center);
		\draw (296.center) to (298.center);
		\draw (298.center) to (300.center);
		\draw (300.center) to (299.center);
		\draw (299.center) to (297.center);
		\draw (297.center) to (296.center);
		\draw (255.center) to (256.center);
		\draw (257.center) to (258.center);
		\draw (259.center) to (260.center);
		\draw (261.center) to (262.center);
		\draw (263.center) to (264.center);
		\draw (265.center) to (266.center);
		\draw (267.center) to (268.center);
		\draw (269.center) to (270.center);
		\draw (271.center) to (272.center);
		\draw (273.center) to (274.center);
		\draw (275.center) to (276.center);
		\draw (277.center) to (278.center);
		\draw (279.center) to (280.center);
		\draw (319.center) to (320.center);
		\draw (321.center) to (322.center);
		\draw (323.center) to (324.center);
	\end{pgfonlayer}
\end{tikzpicture}}
    \caption{Schematic decomposition of the parallel 3-slot purification strategy.  The encoder isometry \(V_{\mathrm{enc}}\) maps the input system \(P\) to the three channel inputs \(I_1,I_2,I_3\) and a one-qubit memory register \(R\).  The three noisy channels \(\cN_\gamma\circ\cU\) are then applied in parallel, producing outputs \(O_1,O_2,O_3\).  The decoder isometry \(V_{\mathrm{dec}}\) processes \(R,O_1,O_2,O_3\) and outputs the final system qubit \(F\), while four ancillary qubits are traced out.}
    \label{fig:par}
\end{figure}

To understand why three channel uses can improve fidelity, recall that each noisy channel $\mathcal N_\gamma\circ\mathcal U$ applies $U$ followed by depolarizing noise. Our encoder exploits the following decomposition of the collective unitary operation:
\begin{equation*}
    U^{\otimes3}\simeq(I_2\otimes U)\oplus f_4(U),
\end{equation*}
where $f_4$ denotes the four-dimensional irreducible representation of $\mathrm{SU}(2)$. Since $I_2\otimes U=U\oplus U$, two orthogonal subspaces carry the same target operation $U$. Our encoder coherently encodes the input in these two subspaces and entangles their label with the memory $R$, as shown in Appendix~\ref{sec:cir}. The three unitary operations apply the same target $U$ in both subspaces, leaving their label and its correlations with $R$ unchanged. The subsequent local noise affects the two encodings differently, and their correlations with $R$ provide a resource for partial recovery by the joint decoder in Fig.~\ref{fig:par}. The two-slot action $U^{\otimes2}$ lacks this pair of equivalent qubit representations.

The fidelity gain over the single-use baseline $1-\frac{3}{4}\gamma$ is positive for $0<\gamma<1$, as shown by the factorization in Appendix~\ref{app:proof_universal_go}. It vanishes linearly at both endpoints: an ideal unitary channel already has unit fidelity, while completely depolarizing channels contain no information about $U$. The largest gain is approximately $0.07334$, attained at $\gamma\simeq0.3713$.

\section{Discussion}
\label{sec:discussion}
For deterministic universal purification of unknown single-qubit unitaries under independent depolarizing noise, two uses give no advantage over a single use, even with ICO. Three uses suffice, and our explicit parallel circuit attains the ICO optimum using a one-qubit memory and fixed encoding and decoding operations.

For more channel uses in the present qubit model, unitary twirling alone leaves exponentially many block parameters before imposing the ICO equations~\cite{grassl1998computing}. Slot-permutation symmetry equates coefficients for different orderings of the same local operator basis elements. Only their occurrence counts need to be distinguished, giving at most polynomially many independent parameters and independent ICO affine equations as the number of channel uses $n$ grows. On the construction side, an open direction is to investigate how smaller protocols can be combined into strategies that improve fidelity with more channel uses.

Parameterizing the channels also reveals connections with quantum metrology and polynomial-transformation methods. To derive their metrological upper bound, Niwa et al.\ consider the one-parameter family $U_t=e^{iHt}$ with fixed $H$, and use limits on estimating $t$ to bound sequential purification fidelity in the low-noise regime~\cite{niwa2026scaling}. In our three-use protocol, the depolarizing contraction $1-\gamma$ serves as the variable of a cubic polynomial transformation, providing a point of comparison with the polynomial signal transformations used in quantum signal processing (QSP)~\cite{low2017optimal}.

The Choi formulation and SDP also provide a starting point for studying other dimensions, target families, and noise models. Our analytic reduction relies on qubit $\mathrm{SU}(2)$ representations and depolarizing covariance, so these extensions require checking the available symmetries and deriving new bounds. Relaxing the deterministic requirement would allow us to study the trade-off between conditional fidelity and success probability.

\section*{Acknowledgments}
This work was partially supported by the National Key R\&D Program of China (Grant No.~2024YFB4504004), the National Natural Science Foundation of China (Grant No.~12447107, 92576114, 62602562), the Guangdong Provincial Quantum Science Strategic Initiative (Grant No.~GDZX2403008, GDZX2503001), the CCF-Tencent Rhino-Bird Open Research Fund, and the Guangdong Provincial Key Lab of Integrated Communication, Sensing and Computation for Ubiquitous Internet of Things (Grant No.~2023B1212010007). 


\newpage
\bibliographystyle{unsrt}
\bibliography{ref}

\clearpage
\appendices

\setcounter{subsection}{0}
\setcounter{table}{0}
\setcounter{figure}{0}


\begin{center}
\large{\textbf{Appendix for
`Distilling Qubit Unitary Operations: A No-Go Theorem and Minimal Realization'}}
\end{center}
\renewcommand{\theequation}{S\arabic{equation}}
\renewcommand{\thesubsection}{\normalsize{Supplementary Note \arabic{subsection}}}
\renewcommand{\theproposition}{S\arabic{proposition}}
\renewcommand{\thedefinition}{S\arabic{definition}}
\renewcommand{\thefigure}{S\arabic{figure}}
\setcounter{equation}{0}
\setcounter{table}{0}
\setcounter{section}{0}
\setcounter{proposition}{0}
\setcounter{definition}{0}
\setcounter{figure}{0}

\renewcommand{\theHequation}{S\arabic{equation}}
\renewcommand{\theHproposition}{S\arabic{proposition}}
\renewcommand{\theHfigure}{S\arabic{figure}}

\section{Proof of Theorem~\ref{thm:nogo}}\label{app:proof_universal_no-go}

We prove Theorem \ref{thm:nogo} by showing that for depolarizing noise with noise level $\gamma$, the optimal channel fidelity after purification is $1-\frac{3}{4}\gamma$, the same as that of a trivial purification.

\begin{theorem}\label{theo:no-go_universal}
For the $2$-dimensional depolarizing noise channel $\cN_{\gamma}$ with noise level $\gamma\in(0,1)$ and any $2$-slot ICO
strategy, we have
\begin{align}
    &\max_{\cC\in\{ICO\}}\mathbb E_{U}C_{P I^2 O^2 F}*(J^{\cN_{\gamma}}*\dketbra{U}{U})^{\otimes 2}*(\dketbra{U}{U})^T/4\nonumber\\
    =&1-\frac{3}{4}\gamma.\label{eq:no-go_max}
\end{align}
\end{theorem}

\begin{proof}
Firstly, we find
\begin{equation}
    J^{\cN_{\gamma}}*\dketbra{U}{U}=(1-\gamma)\dketbra{U}{U}+\frac{\gamma}{2}I_4
\end{equation}
and
\begin{equation}
    J^{\cN_{\gamma}}*\dketbra{U}{U}*(\dketbra{U}{U})^T=4-3\gamma.
\end{equation}
Consider a $2$-slot general ICO strategy $\cC$. $C_{P I^2 O^2 F}*(J^{\cN_{\gamma}}*\dketbra{U}{U})^{\ox2}$ is a qubit channel and have trace $d=2$, so \eqref{eq:no-go_max} is equivalent to
\begin{align}
    &\max_{\cC}\mathbb E_{U}
    C_{P I^2 O^2 F}*(J^{\cN_{\gamma}}*\dketbra{U}{U})^{\otimes 2}*(\dketbra{U}{U})^T\nonumber\\
    -&\frac{4-3\gamma}{2}\tr\left[C_{P I^2 O^2 F}*(J^{\cN_{\gamma}}*\dketbra{U}{U})^{\ox2}\right]
    =0.\label{eq:no-go_max_2}
\end{align}

After rewriting the left hand of \eqref{eq:no-go_max_2} as
\begin{align*}
    &\max_{\cC}\mathbb E_{U}C_{P I^2 O^2 F}*(J^{\cN_{\gamma}}*\dketbra{U}{U})^{\otimes 2}*(\dketbra{U}{U})^T\\
    -&\frac{4-3\gamma}{2}\tr\left[C_{P I^2 O^2 F}*(J^{\cN_{\gamma}}*\dketbra{U}{U})^{\ox2}\right]\\
    =&\max_{\cC}\mathbb E_{U}C_{P I^2 O^2 F}*((1-\gamma)\dketbra{U}{U}+\frac{\gamma}{2}I_4)^{\otimes 2}*((\dketbra{U}{U})^T\\
    -&\frac{4-3\gamma}{2}I_4)\\
    =&\max_{\cC}\mathbb E_{U}\tr[C_{P I^2 O^2 F}\cdot(\left((1-\gamma)\dketbra{U}{U}+\frac{\gamma}{2}I_4\right)^{T\otimes 2}_{I^2 O^2}\\
    \ox&\left(\dketbra{U}{U}-\frac{4-3\gamma}{2}I_4\right)_{PF})]\\
    =&\max_{\cC}\mathbb E_{U}\tr[C_{P I^2 O^2 F}\cdot(\left((1-\gamma)\dketbra{U^*}{U^*}+\frac{\gamma}{2}I_4\right)^{\otimes 2}_{I^2 O^2}\\
    \ox&\left(\dketbra{U}{U}-\frac{4-3\gamma}{2}I_4\right)_{PF})]\\
    =&\max_{\cC}\tr[C_{P I^2 O^2 F}\cdot\mathbb E_{U}[\left((1-\gamma)\dketbra{U^*}{U^*}+\frac{\gamma}{2}I_4\right)^{\otimes 2}_{I^2 O^2}\\
    \ox&\left(\dketbra{U}{U}-\frac{4-3\gamma}{2}I_4\right)_{PF}]],
\end{align*}
we have another equivalent equation of \eqref{eq:no-go_max}:
\begin{align}
    &\max_{\cC}\tr[C_{P I^2 O^2 F}\cdot\mathbb E_{U}[\left((1-\gamma)\dketbra{U^*}{U^*}+\frac{\gamma}{2}I_4\right)^{\otimes 2}_{I^2 O^2}\nonumber\\
    \ox&\left(\dketbra{U}{U}-\frac{4-3\gamma}{2}I_4\right)_{PF}]]=0.\label{eq:no-go_max_3}
\end{align}
{\bf Suppose} that there exists a $2$-slot general strategy $\check{\cC}$ that satisfies
\begin{align}
    &\tr[\check{C}_{P I^2 O^2 F}\cdot\mathbb E_{U}[\left((1-\gamma)\dketbra{U^*}{U^*}+\frac{\gamma}{2}I_4\right)^{\otimes 2}_{I^2 O^2}\nonumber\\
    \ox&\left(\dketbra{U}{U}-\frac{4-3\gamma}{2}I_4\right)_{PF}]]>0.\label{eq:asumpt}
\end{align}
Notice that for any $V,W\in\operatorname{SU}(2)$, the two-slot unitary $U_{V,W}$ defined in the proof sketch of Theorem~\ref{thm:nogo} and 
\begin{align*}
    &\mathbb E_{U}\Bigl[\left((1-\gamma)\dketbra{U^*}{U^*}+\frac{\gamma}{2}I_4\right)^{\otimes 2}_{I^2 O^2}\\
    &\qquad\ox\left(\dketbra{U}{U}-\frac{4-3\gamma}{2}I_4\right)_{PF}\Bigr]
\end{align*}
commute. Then \eqref{eq:asumpt} still holds after replacing $\check{C}_{P I^2 O^2 F}$ with 
\begin{equation}
    U_{V,W}\check{C}_{P I^2 O^2 F}U_{V,W}^\dagger
\end{equation}
for any $V,W\in\operatorname{SU}(2)$, or their expectation
\begin{align*}
    &\bar{C}_{P I^2 O^2 F}\\
    :=&\mathbb E_{V,W}(V_P\ox V^{*\ox2}_{I_1 I_2}\ox W^{*\ox2}_{O_1 O_2}\ox W_F)\\
    &\cdot \check{C}_{P I^2 O^2 F}\cdot(V_P^\dag\ox V^{T\ox2}_{I_1 I_2}\ox W^{T\ox2}_{O_1 O_2}\ox W^\dag_F).
\end{align*}
Thus we have obtained $\bar{C}_{P I^2 O^2 F}$ that satisfies
\begin{equation}\label{eq:asumpt_1}
\begin{cases}
    \begin{aligned}
        &\tr[\bar{C}_{P I^2 O^2 F}\cdot\mathbb E_{U}[\left((1-\gamma)\dketbra{U^*}{U^*}+\frac{\gamma}{2}I_4\right)^{\otimes 2}_{I^2 O^2}\\
        \ox&\left(\dketbra{U}{U}-\frac{4-3\gamma}{2}I_4\right)_{PF}]]>0
    \end{aligned}\\
    \left[V_P\ox V^{*\ox2}_{I_1 I_2}\ox W^{*\ox2}_{O_1 O_2}\ox W_F,\bar{C}_{P I^2 O^2 F}\right]=0
\end{cases}
\end{equation}
for all $V,W\in\operatorname{SU}(2)$. By the irreducible representation decomposition of 
\begin{equation}
    \operatorname{SU}(2)\rightarrow \operatorname{SU}(8):\ U\mapsto U^{*\ox 2}\ox U=2U\oplus f_{4}(U),
\end{equation}
there exist $G\in\operatorname{SU}(8)$ such that for all $U\in\operatorname{SU}(2)$
\begin{equation}
    \ G(U^{*\ox 2}\ox U)G^\dag=
    \begin{pmatrix}
        I_2\ox U\\&f_{4}(U)
    \end{pmatrix},
\end{equation}
where $f_4$ is a $4$-dimensional irreducible representation of $\operatorname{SU}(2)$. Moreover by the irreducible representation decomposition of $\operatorname{SU}(2) \times \operatorname{SU}(2) \rightarrow \operatorname{SU}(64)$, we have:
\begin{align*}
    (V,W) \mapsto &V\ox V^{*\ox 2}\ox W^{*\ox2}\ox W\\
    =&4V\ox W \oplus 2f_{4}(V)\ox W\\
    &\quad\oplus 2f_{4}(W)\ox V\oplus f_{4}(V)\ox f_4(W).
\end{align*}

To express the above representation decomposition in the port order $(P,I_1,O_1,I_2,O_2,F)$, we use the fixed Schur-basis transformation $\tilde G$ provided in our repository\footnote{The explicit unitary $\tilde G$ can be found at our GitHub repository \url{https://github.com/jyzhau/Unitary-Purification}.}. All block matrices below are evaluated in this same basis.
Since $ \bar{C}_{P I^2 O^2 F}$ is always commutative with 
$V_P\ox V^{*\ox2}_{I_1 I_2}\ox W^{*\ox2}_{O_1 O_2}\ox W_F$
for any $V,W\in\operatorname{SU}(2)$, 
by Schur's lemma, $\bar{C}_{P I^2 O^2 F}$ takes the form
\begin{align*} 
    \tilde G\bar{C}_{P I^2 O^2 F}\tilde G^\dag
    =\begin{pmatrix}
        H_0\ox I_4\\&H_1\ox I_8\\&&H_2\ox I_8\\&&&h_3I_{16}
    \end{pmatrix}\\
    =\begin{pmatrix}
        H_0\\&H_1\ox I_2\\&&H_2\ox I_2\\&&&h_3I_{4}
    \end{pmatrix}\ox I_4,
\end{align*}
where $H_0=\{H_{0jk}\}_{jk}\in \mathbb C^{4\times4}$, $H_1=\{H_{1jk}\}_{jk},H_2=\{H_{2jk}\}_{jk}\in \mathbb C^{2\times2}$, $h_3\in \mathbb C$.
Denoting $\tr_d$ as the partial trace on the last $d$-dimensional system, and 
\begin{align*}
    &\tr_4[\tilde G\cdot\mathbb E_{U}[\left((1-\gamma)\dketbra{U^*}{U^*}+\frac{\gamma}{2}I_4\right)^{\otimes 2}_{I^2 O^2}\\
    \ox&\left(\dketbra{U}{U}-\frac{4-3\gamma}{2}I_4\right)_{PF}]\cdot\tilde G^\dag]\\
    =&\sum_{jk=0}^3\ketbra{j}{k}\ox \Omega_{jk},
\end{align*}
we have
\begin{align*}
    &\tr[\bar{C}_{P I^2 O^2 F}\cdot\mathbb E_{U}[\left((1-\gamma)\dketbra{U^*}{U^*}+\frac{\gamma}{2}I_4\right)^{\otimes 2}_{I^2 O^2}\\
    \ox&\left(\dketbra{U}{U}-\frac{4-3\gamma}{2}I_4\right)_{PF}]] \\
    =&\tr[H_0\Omega_{00}]+\sum_{j=1}^{2}\tr[H_j\tr_2[\Omega_{jj}]]+h_3\tr[\Omega_{33}],
\end{align*}
%
where the required blocks are explicitly
\begin{align*}
    \Omega_{00}={}&-(1-\gamma)
    \begin{pmatrix}
        \frac{3(1-\gamma)^2+1}{6} & 2-\gamma\\
        2-\gamma & \frac{3[3(1-\gamma)^2+1]}{2}
    \end{pmatrix}\\
    &\oplus\gamma(1-\gamma)
    \begin{pmatrix}
        -\frac32(2-\gamma) & 1\\
        1 & -\frac32(2-\gamma)
    \end{pmatrix},
\end{align*}
\begin{gather*}
    \tr_2[\Omega_{11}]=\tr_2[\Omega_{22}]\\
    =-(1-\gamma)\operatorname{diag}\!\left(
        \frac{3(1-\gamma)^2+13}{3},\,3\gamma(2-\gamma)\right),\\
    \tr[\Omega_{33}]=-\frac23(1-\gamma)[3(1-\gamma)^2+7].
\end{gather*}

An indefinite causal order (ICO) strategy $\bar{\cC}$ is valid if and only if its corresponding Choi operator $\bar{C}$ on the joint system $P I^2 O^2 F$ satisfies the following conditions:
\begin{equation}
\begin{cases}
    \bar{C} \ge 0 \\
    {}_{I_1 O_1 F} \bar{C} = {}_{I_1 O_1 O_2 F} \bar{C} \\
    {}_{I_2 O_2 F} \bar{C} = {}_{O_1 I_2 O_2 F} \bar{C} \\
    {}_{I^2 O^2 F} \bar{C} = {}_{P I^2 O^2 F} \bar{C} \\
    {}_{F} \bar{C} + {}_{O_1 O_2 F} \bar{C} = {}_{O_1 F} \bar{C} + {}_{O_2 F} \bar{C} \\
    \tr[\bar{C}] = 8
\end{cases}
\end{equation}
Substituting the block form of $\bar C$ into these conditions and comparing matrix entries gives
\begin{equation}
\begin{cases}
    H_0,H_1,H_2,h_3\ge0\\   
    H_{100}=\frac{1}{2} \left(3 H_{033}-H_{000}\right)\\
    H_{101}=\frac{1}{2} \left(3 H_{013}^*-H_{002}\right)\\
    H_{111}=\frac{1}{2} \left(3 H_{011}-H_{022}\right)\\
    H_{211}=\frac{1}{4} \left(-2 H_{011}-2 H_{033}+1\right)\\
    h_3=\frac{1}{8} \left(-6 H_{011}-6 H_{033}-4 H_{200}+3\right)
\end{cases}
\end{equation}
Substituting these relations into the block-trace objective above gives
\begin{equation}\label{eq:sdp_7+}
\begin{aligned}
    \max\; & (1-\gamma)\tr\left[\begin{pmatrix}
         H_{000} & H_{001} \\H_{001}^* & H_{011}
    \end{pmatrix}\begin{pmatrix}
        2 & -(2-\gamma) \\ -(2-\gamma) & -1
    \end{pmatrix}\right]\\
    +&(1-\gamma)\tr\left[\begin{pmatrix}
         H_{022} & H_{023} \\ H_{023}^* & H_{033}
    \end{pmatrix}\begin{pmatrix}
        0 & \gamma \\ \gamma & -3
    \end{pmatrix}\right]\\
    -&\frac{1}{2}(1-\gamma) \left(4 H_{200}+5\right)\\
    \text{s.t.}\; &H_0\ge0\\  
    &\begin{pmatrix}
         \frac{1}{2} \left(3 H_{033}-H_{000}\right) & \frac{1}{2} \left(3 H_{013}^*-H_{002}\right) \\
         \frac{1}{2} \left(3 H_{013}^*-H_{002}\right)^* & \frac{1}{2} \left(3 H_{011}-H_{022}\right)
    \end{pmatrix}\ge0\\
    &\begin{pmatrix}
         H_{200} & H_{201} \\ 
         H_{201}^* & \frac{1}{4} \left(-2 H_{011}-2H_{033}+1\right)
    \end{pmatrix}\ge0\\
    &-6 H_{011}-6 H_{033}-4 H_{200}+3\ge0.
\end{aligned}
\end{equation}
For any optimal solution of \eqref{eq:sdp_7+}, maintaining the values of $\{H_{000},\allowbreak H_{001},\allowbreak H_{011},\allowbreak H_{022},\allowbreak H_{023},\allowbreak H_{033},\allowbreak H_{200}\}$ and replacing the values of any other variables (that is, in $\{H_{002},\allowbreak H_{003},\allowbreak H_{012},\allowbreak H_{013},\allowbreak H_{201}\}$) by $0$, we will obtain another optimal solution because it is an easily checked feasible solution with the same objective function value with the optimal solution of the origin. 
The new optimal solution could be described by the following SDP:
\begin{equation}\label{eq:sdp_7}
\begin{aligned}
    \max\; & (1-\gamma)\tr\left[\begin{pmatrix}
         H_{000} & H_{001} \\ H_{001}^* & H_{011}
    \end{pmatrix}\begin{pmatrix}
        2 & -(2-\gamma) \\
        -(2-\gamma) & -1  \\
    \end{pmatrix}\right]\\
    +&(1-\gamma)\tr\left[\begin{pmatrix}
         H_{022} & H_{023} \\ H_{023}^* & H_{033}
    \end{pmatrix}\begin{pmatrix}
        0 & \gamma \\ \gamma & -3  \\
    \end{pmatrix}\right]\\
    -&\frac{1}{2}(1-\gamma) \left(4 H_{200}+5\right)\\
    \text{s.t.}\; 
    &\begin{pmatrix}
         H_{000} & H_{001} \\ H_{001}^* & H_{011}
    \end{pmatrix},\begin{pmatrix}
         H_{022} & H_{023} \\ H_{023}^* & H_{033}
    \end{pmatrix}\ge0\\  
    &3 H_{033}-H_{000},3 H_{011}-H_{022}\ge0\\
    &H_{200},-2 H_{011}-2H_{033}+1\ge0\\
    &-6 H_{011}-6 H_{033}-4 H_{200}+3\ge0.
    \end{aligned}
\end{equation}
Firstly, we could find the above SDP reaches its maximum only if $H_{200}=0$, and simplify it as
 
\begin{equation}\label{eq:sdp_6}
    \begin{aligned}
    \max\; & \tr[\left(
        \begin{array}{cc}
         H_{000} & H_{001} \\
         H_{001}^* & H_{011}
        \end{array}
    \right)\\
    \cdot&\left(
\begin{array}{cc}
 2 (1-\gamma)  & -(1-\gamma)  (2-\gamma) \\
 -(1-\gamma)  (2-\gamma) & -(1-\gamma)  \\
\end{array}
\right)]\\
    +&\tr\left[\left(
        \begin{array}{cc}
         H_{022} & H_{023} \\
         H_{023}^* & H_{033}
        \end{array}
    \right)\left(
\begin{array}{cc}
 0 & (1-\gamma)\gamma \\
 (1-\gamma)\gamma & -3 (1-\gamma)  \\
\end{array}
\right)\right]\\
-&\frac{5}{2}(1-\gamma)\\
    \text{s.t.}\; &\left(
        \begin{array}{cc}
         H_{000} & H_{001} \\
         H_{001}^* & H_{011}
        \end{array}
    \right),\left(
        \begin{array}{cc}
         H_{022} & H_{023} \\
         H_{023}^* & H_{033}
        \end{array}
    \right)\ge0\\  
    &3 H_{033}-H_{000},3 H_{011}-H_{022},-2 H_{011}-2H_{033}+1\ge0.
    \end{aligned}
\end{equation}
Secondly, for fixed diagonal entries, positivity gives
$|H_{001}|^2\le H_{000}H_{011}$ and $|H_{023}|^2\le H_{022}H_{033}$.
For $0<\gamma<1$, their real parts have negative and positive objective coefficients, respectively.
Thus the objective is maximized by
$H_{001}=-\sqrt{H_{000}H_{011}}$ and $H_{023}=\sqrt{H_{022}H_{033}}$, which remain feasible.
After substitution, for every feasible pair $(H_{011},H_{033})$, the objective is increasing in $H_{000}$ and nondecreasing in $H_{022}$.
Hence we may take $H_{000}=3H_{033}$ and $H_{022}=3H_{011}$ without loss of optimality.

Since the matrix $M$ in Eq.~\eqref{eq:matrix_M} has maximal eigenvalue $5(1-\gamma)$
with eigenstate $(\frac{1}{2},\frac{\sqrt{3}}{2})^T$,
the quadratic objective obtained after the above substitutions attains its maximum $0$ if and only if
$(\sqrt{H_{011}},\sqrt{H_{033}})=\frac{1}{\sqrt{2}}(\frac{1}{2},\frac{\sqrt{3}}{2})$.

As a result, the optimal value of original SDP is $0$, which contradicts the existence of $\bar{\cC}$ in \eqref{eq:asumpt_1}. Thus, the origin assumption \eqref{eq:asumpt} is false, which leads to
\begin{equation}
    F_{\max}^{\mathrm{ICO}}(2;\cN_\gamma)\le1-\frac{3}{4}\gamma.
\end{equation}

It remains to connect this average bound to the pointwise universal
condition in Definition~\ref{def:purification_protocol}.  Suppose that a
$2$-slot ICO strategy satisfies that condition and is strictly better for
some $U_0\in\operatorname{SU}(2)$, so that
\begin{equation}
    F\!\left(J^{\mathcal E_{U_0}}_{PF},\dketbra{U_0}{U_0}\right)
    >1-\frac{3}{4}\gamma.
\end{equation}
For a fixed strategy $C$, the map
\begin{equation}
    U\longmapsto
    F\!\left(J^{\mathcal E_U}_{PF},\dketbra{U}{U}\right)
\end{equation}
is continuous.  Indeed,
$U\mapsto J^{\cN_\gamma\circ\cU}$ is continuous, the link product is
linear, and channel fidelity is continuous in the resulting
finite-dimensional Choi operator.  Hence the strict inequality above
continues to hold throughout a nonempty open neighborhood of $U_0$.

The universal condition guarantees that the same fidelity is at least
$1-\frac{3}{4}\gamma$ for every other $U$.  Since every nonempty open
subset of $\operatorname{SU}(2)$ has positive normalized Haar measure, we
obtain
\begin{equation}
    \int_{\operatorname{SU}(2)}
    F\!\left(J^{\mathcal E_U}_{PF},\dketbra{U}{U}\right)\,\mathrm dU
    >1-\frac{3}{4}\gamma.
\end{equation}
This contradicts the average upper bound proved above.  Therefore, any
$2$-slot ICO strategy satisfying the universal condition must attain
exactly $1-\frac{3}{4}\gamma$ for every $U$ and cannot be strictly better
for even one unitary.  This proves the pointwise statement in
Theorem~\ref{thm:nogo}.
\end{proof}

\section{Proof of Theorem~\ref{thm:go}}\label{app:proof_universal_go}

\begin{theorem}[Optimal $3$-slot purification]
\label{thm:go2}
For all $\gamma\in(0,1)$, the optimal average fidelity over all $3$-slot ICO strategies under the depolarizing channel $\cN_\gamma$ is
\begin{equation}\label{eq:fideopt}
\begin{aligned}
    F_{\max}^{\mathrm{ICO}}(3;\cN_\gamma)
    &=\frac{2-\gamma}{24}\bigl[12+(8\sqrt2-9)\gamma\\
    &\hspace{29mm}+(3-8\sqrt2)\gamma^2\bigr].
\end{aligned}
\end{equation}
This optimum is attained by the explicit $3$-slot parallel strategy constructed in Appendix~\ref{sec:cir}, which satisfies the universal purification condition~\eqref{eq:purification_condition} with strict inequality for every $U\in\mathrm{SU}(2)$. Consequently, the optimal average fidelities coincide: $F_{\max}^{\mathrm{Par}}(3;\cN_\gamma)=F_{\max}^{\mathrm{Seq}}(3;\cN_\gamma)=F_{\max}^{\mathrm{ICO}}(3;\cN_\gamma)$.
\end{theorem}

\begin{proof}
First, we prove the upper bound in Eq.~\eqref{eq:fideopt}. Consider an arbitrary $3$-slot ICO strategy $\check{\cC}$ with Choi operator $\check{C}$. All eight ports are qubits, ordered as $P,I_1,O_1,I_2,O_2,I_3,O_3,F$. Write $D=I_4/2$ for the Choi operator of the completely depolarizing qubit channel.
The average fidelity is $F_{\mathrm{ave}}(\cN_\gamma,\check{\cC})=\operatorname{Tr}(\Omega\check{C})$, where
\begin{equation}
\begin{aligned}
    \Omega=\frac14\mathbb E_U\Bigl[
    &\bigl((1-\gamma)\dketbra{U^*}{U^*}+\gamma D\bigr)^{\otimes3}_{I^3O^3}\\
    &\otimes(\dketbra{U}{U})_{PF}\Bigr].
\end{aligned}
\end{equation}
The factor $1/4$ is the qubit channel-fidelity normalization in Eq.~\eqref{eq:avg_fidelity}. Each tensor factor acts on its labelled ports in the order specified above.

As in the proof of Theorem~\ref{thm:nogo}, $\Omega$ commutes with $U_{V,W}=V_P\otimes V^{*\otimes3}_{I^3}\otimes W^{*\otimes3}_{O^3}\otimes W_F$ for every $V,W\in\mathrm{SU}(2)$. It is also invariant under permutations of the three identical slots. We can therefore replace $\check{C}$ by its average
\begin{equation}\label{eq:twirl_three}
    \bar{C}=\frac16\sum_{\pi\in S_3}\mathbb E_{V,W}
    \bigl[P_\pi U_{V,W}\check{C}U_{V,W}^\dagger P_\pi^\dagger\bigr]
\end{equation}
without changing its fidelity. Here $P_\pi$ permutes the complete slots $(I_j,O_j)$.

To check that this average remains an ICO strategy, specialize Eq.~\eqref{eq:ico1} to three slots:
\begin{equation}\label{eq:full_constraints}
\begin{cases}
    C\succeq0,\\
    \displaystyle\sum_{S'\subseteq S}(-1)^{|S'|}
    {}_{I_{\bar S}O_{\bar S}O_{S'}F}C=0
    \quad(\emptyset\ne S\subseteq\{1,2,3\}),\\
    {}_{I^3O^3F}C={}_{PI^3O^3F}C,\\
    \operatorname{Tr}C=16.
\end{cases}
\end{equation}
Here $\bar S=\{1,2,3\}\setminus S$, and replaced identity factors are restored to their original port positions. Conjugation by $U_{V,W}$ commutes with each normalized trace-and-replace operation, while a slot permutation maps the equation for $S$ to the equation for $\pi(S)$. Thus every operator being averaged satisfies the same ICO conditions. Averaging preserves positive semidefiniteness, all linear equalities, and the trace, so $\bar{C}$ is a valid ICO strategy with the same fidelity as $\check{C}$. Moreover,
\begin{equation}
    [\bar{C},U_{V,W}]=0,\qquad
    P_\pi\bar{C}P_\pi^\dagger=\bar{C}
    \quad(\pi\in S_3).
\end{equation}

Next, we apply the Clebsch--Gordan decomposition. By Schur's lemma, there exists a unitary $G$ such that
\begin{equation}
    G^\dagger\bar{C}G
    =\bigoplus_{k=0}^{8}(I_{d_k}\otimes H_k),
    \qquad H_k\succeq0.
\end{equation}
The sizes of $H_0,\ldots,H_8$ are $4,6,2,6,9,3,2,3,1$, respectively, with $H_8=[h_8]$, and the irreducible representation dimensions are
\begin{equation*}
    (d_0,\ldots,d_8)=(1,3,5,3,9,15,5,15,25).
\end{equation*}
Substituting this block form into Eq.~\eqref{eq:full_constraints} and imposing the two generating slot exchanges $P_{12},P_{23}$ gives 166 independent affine equalities on the 196 real Hermitian block coordinates, leaving 30 free parameters.

To evaluate the objective, recall that each noisy channel is a mixture of the ideal unitary channel, with weight $1-\gamma$, and the completely depolarizing channel, with weight $\gamma$. In the product of three channel uses, collect the contributions in which exactly $r$ slots contain the completely depolarizing channel, where $r=0,1,2,3$. Let $\Omega_k^{(r)}$ be the $k$-th reduced CG block of the performance operator for this contribution, averaged uniformly over the $\binom3r$ possible positions. These blocks include the fidelity factor $1/4$ but exclude the mixture weight $\binom3r\gamma^r(1-\gamma)^{3-r}$. In particular, averaging over positions gives a factor $1/12$ for $r=1,2$. Writing $F_{\mathrm{ave}}(\bar{C})$ for the average fidelity of the symmetrized strategy, we obtain
\begin{equation}\label{eq:noise_count_fidelity}
\begin{aligned}
    F_{\mathrm{ave}}(\bar{C})
    =\sum_{r=0}^{3}&\binom3r\gamma^r(1-\gamma)^{3-r}\\
    &\times\sum_{k=0}^{8}d_k\operatorname{Tr}(\Omega_k^{(r)}H_k).
\end{aligned}
\end{equation}
Thus each inner sum is the average fidelity obtained by inserting exactly $r$ completely depolarizing channels and $3-r$ ideal unitary channels into the same strategy.

The cases $r=0$ and $r=3$ are immediate. When all three slots contain the ideal unitary channel, the output is a valid channel whose fidelity with the target is at most one. When all three contain the completely depolarizing channel, the averaged performance operator is $I_{256}/64$, giving $\operatorname{Tr}(\bar{C})/64=1/4$.

For $r=1,2$, we prove the respective upper bounds
\begin{equation}
    \beta_1=\frac{21+8\sqrt2}{36},\qquad
    \beta_2=\frac{27+8\sqrt2}{72}.
\end{equation}
We construct positive semidefinite matrices $Q_k^{(r)}$ such that, for every collection of blocks satisfying the ICO and $S_3$ equalities,
\begin{equation}\label{eq:full_remainder_identity}
\begin{aligned}
    &\beta_r-\sum_{k=0}^{8}d_k\operatorname{Tr}(\Omega_k^{(r)}H_k)\\
    &\qquad=\sum_{k=0}^{8}\operatorname{Tr}(Q_k^{(r)}H_k),
    \qquad r=1,2.
\end{aligned}
\end{equation}
Substituting the affine solution into the difference between the two sides makes the constant and every free-parameter coefficient vanish. The positive semidefiniteness of $Q_k^{(r)}$ is also verified exactly. Since $H_k\succeq0$, each trace on the right is nonnegative, proving the bounds $\beta_1$ and $\beta_2$. The explicit matrices, the CG transformation, and the complete linear equalities are provided in our GitHub repository\footnote{\url{https://github.com/jyzhau/Unitary-Purification}.}.

Combining these bounds with the nonnegative mixture weights gives
\begin{equation}
\begin{aligned}
    F_{\mathrm{ave}}(\bar{C})
    &\le(1-\gamma)^3+3\gamma(1-\gamma)^2\beta_1\\
    &\quad+3\gamma^2(1-\gamma)\beta_2+\frac{\gamma^3}{4}\\
    &=\frac{2-\gamma}{24}\bigl[12+(8\sqrt2-9)\gamma\\
    &\hspace{34mm}+(3-8\sqrt2)\gamma^2\bigr].
\end{aligned}
\end{equation}
This proves the upper bound over all three-slot ICO strategies. For achievability, the parallel encoder--decoder in Appendix~\ref{sec:cir} attains this value. The inclusions $\mathcal S_3^{\mathrm{Par}}\subseteq\mathcal S_3^{\mathrm{Seq}}\subseteq\mathcal S_3^{\mathrm{ICO}}$ then show that the optimal fidelities of all three classes coincide.

Finally, the parallel construction is covariant, so its fidelity is independent of the unknown unitary. Its gain over the bare noisy channel is
\begin{equation*}
\begin{aligned}
    &F_{\max}^{\mathrm{ICO}}(3;\cN_\gamma)
      -\left(1-\frac{3\gamma}{4}\right)\\
    &\quad=\frac{\gamma(1-\gamma)}{24}
    \bigl[4(4\sqrt2-3)(1-\gamma)+(8\sqrt2-9)\gamma\bigr]>0
\end{aligned}
\end{equation*}
for $0<\gamma<1$, proving nontrivial universal purification.
\end{proof}

\section{The Circuit Decomposition of 3-Slot Unitary Purification Protocol}\label{sec:cir}

This appendix gives the gate-level realization of the parallel 3-slot strategy used in the achievability part of Theorem~\ref{thm:go}.  The notation follows the schematic in Fig.~\ref{fig:par}: the encoder isometry
\begin{equation}
    V_{\mathrm{enc}}:\mathcal{H}_{P}\longrightarrow
    \mathcal{H}_{I_1}\otimes\mathcal{H}_{I_2}\otimes\mathcal{H}_{I_3}\otimes\mathcal{H}_{R}
\end{equation}
maps the input qubit $P$ to the three channel inputs $I_1,I_2,I_3$ and a one-qubit memory register $R$.

In the basis order $(I_1,I_2,I_3)$, define
\begin{equation*}
\begin{aligned}
    A\ket{\psi}&=\frac{\ket{01}-\ket{10}}{\sqrt2}\otimes\ket{\psi},\\
    B\ket{0}&=\frac{2\ket{001}-\ket{010}-\ket{100}}{\sqrt6},\\
    B\ket{1}&=\frac{\ket{011}+\ket{101}-2\ket{110}}{\sqrt6}.
\end{aligned}
\end{equation*}
The encoder then takes the form
\begin{equation*}
\begin{aligned}
    V_{\mathrm{enc}}\ket{\psi}
    &=\frac{A\ket{\psi}\otimes\ket{0}_R
       -B\ket{\psi}\otimes\ket{1}_R}{\sqrt2},\\
    A^\dagger A&=B^\dagger B=I_2,\qquad A^\dagger B=0,\\
    U^{\otimes3}A&=AU,\qquad U^{\otimes3}B=BU,
\end{aligned}
\end{equation*}
for every $U\in\mathrm{SU}(2)$. Thus $A,B$ identify two orthogonal copies of the qubit representation. The encoder entangles their multiplicity label with $R$, while the representation space carries the input. The ideal channels preserve this encoded form and replace $\ket{\psi}$ by $U\ket{\psi}$.

After the three noisy channels act in parallel, the decoder isometry
\begin{equation}
\begin{aligned}
    V_{\mathrm{dec}}&:\mathcal{H}_{R}\otimes\mathcal{H}_{O_1}\otimes\mathcal{H}_{O_2}\otimes\mathcal{H}_{O_3}
    \\
    &\longrightarrow
    \mathcal{H}_{F}\otimes\bigotimes_{j=1}^{4}\mathcal{H}_{\mathrm{Ancilla}_j}
\end{aligned}
\end{equation}
outputs the purified qubit $F$, while the four ancillary output systems $\mathrm{Ancilla}_1,\mathrm{Ancilla}_2,\mathrm{Ancilla}_3,\mathrm{Ancilla}_4$ are traced out.

The circuits below are drawn as unitary circuits on a larger register, but only fixed-ancilla columns are used to define the isometries.  For the encoder circuit in Fig.~\ref{fig:parallel_encoder_circuit}, the three channel-input wires are initialized as
\begin{equation}
    \ket{0}_{I_1}\ket{0}_{I_2}\ket{0}_{I_3}.
\end{equation}
The top wire is the input $P$ and becomes the memory qubit $R$ at the encoder output.  The displayed top-to-bottom wire order is
\begin{equation}
    P\to R,\quad I_3,\quad I_2,\quad I_1,
\end{equation}
which is only a drawing convention; the slot labels are the same $I_1,I_2,I_3$ as in Fig.~\ref{fig:par}.

\begin{figure*}[!t]
    \centering
    \includegraphics[width=0.68\textwidth]{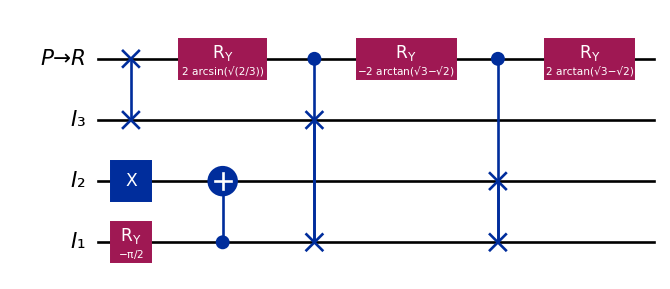}
    \caption{Gate-level encoder circuit for the parallel 3-slot strategy.  The circuit realizes $V_{\mathrm{enc}}$ by applying the displayed four-qubit unitary to the input qubit $P$ together with three fixed ancillary inputs $\ket{0}_{I_1}\ket{0}_{I_2}\ket{0}_{I_3}$.  The output systems are the memory qubit $R$ and the three channel inputs $I_1,I_2,I_3$.}
    \label{fig:parallel_encoder_circuit}
\end{figure*}

For the decoder circuit in Fig.~\ref{fig:parallel_decoder_circuit}, the input register consists of $O_1,O_2,O_3$ together with the memory qubit $R$, and the final-output wire is initialized as $\ket{0}_{F}$.
The circuit is displayed in the top-to-bottom order
\begin{equation}
\begin{aligned}
    O_3&\to \mathrm{Ancilla}_4,\quad O_2\to \mathrm{Ancilla}_3,\\
    O_1&\to \mathrm{Ancilla}_2,\quad R\to \mathrm{Ancilla}_1,
    \quad F.
\end{aligned}
\end{equation}
Thus the first four wires are the decoder inputs, while their outputs are precisely the four ancillary systems that are discarded after decoding.  The bottom wire carries the initialized qubit $F$ and becomes the final output system of the protocol.

\begin{figure*}[!t]
    \centering
    \includegraphics[width=0.97\textwidth]{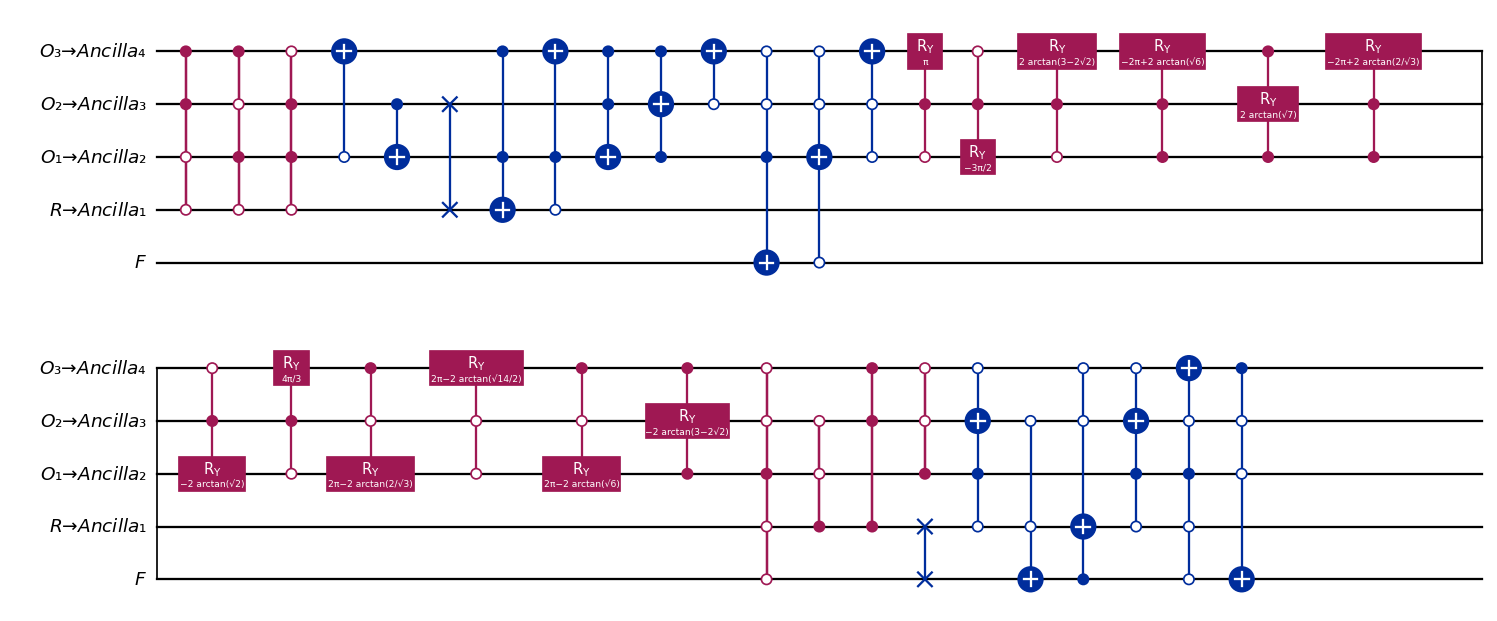}
    \caption{Gate-level decoder circuit for the parallel 3-slot strategy.  The circuit realizes $V_{\mathrm{dec}}$ by applying the displayed five-qubit unitary with the final-output wire initialized to $\ket{0}_{F}$.  At the output, $F$ is retained as the purified system qubit, while $\mathrm{Ancilla}_1,\mathrm{Ancilla}_2,\mathrm{Ancilla}_3,\mathrm{Ancilla}_4$ are traced out.}
    \label{fig:parallel_decoder_circuit}
\end{figure*}

Open control dots in Figs.~\ref{fig:parallel_encoder_circuit} and~\ref{fig:parallel_decoder_circuit} denote control on $\ket{0}$, while filled control dots denote control on $\ket{1}$.  All rotation labels are exact analytic angles.  Extracting the encoder columns corresponding to $\ket{0}_{I_1}\ket{0}_{I_2}\ket{0}_{I_3}$, extracting the decoder columns corresponding to $\ket{0}_{F}$, and tracing out $\mathrm{Ancilla}_1,\mathrm{Ancilla}_2,\mathrm{Ancilla}_3,\mathrm{Ancilla}_4$ gives the Choi comb $C_{P I^3 O^3 F}$ used in the achievability proof.

To evaluate its fidelity, the encoder relations above give $(U^{\otimes3}\otimes I_R)V_{\mathrm{enc}}=V_{\mathrm{enc}}U$, hence $\mathcal E_U=\mathcal E_I\circ\mathcal U$. Here $\mathcal E_I$ is the same circuit with $U=I$ and the noise unchanged. Composing both output and target channels on the input with $\mathcal U^{-1}$ preserves their fidelity, so it suffices to evaluate $U=I$. Substituting this circuit's $C$ into Eq.~\eqref{eq:output_choi}, expanding each noisy insertion as $(1-\gamma)\dketbra{I}{I}+\gamma I_4/2$, and evaluating $\frac14\tr[J^{\mathcal E_I}\dketbra{I}{I}]$ gives
\[
1,\qquad \frac{21+8\sqrt2}{36},\qquad
\frac{27+8\sqrt2}{72},\qquad \frac14
\]
for zero, one, two, and three completely depolarizing insertions, respectively, with each value averaged over their possible positions. The middle values equal $\beta_1,\beta_2$ in Appendix~\ref{app:proof_universal_go}. Combining the mixture weights yields
\begin{equation*}
\begin{aligned}
F_{\mathrm{ave}}(\cN_\gamma,\cC)
&=(1-\gamma)^3+3\gamma(1-\gamma)^2\beta_1\\
&\quad+3\gamma^2(1-\gamma)\beta_2+\frac{\gamma^3}{4},
\end{aligned}
\end{equation*}
which equals Eq.~\eqref{eq:fidelity} and holds for every target $U$.

\vfill

\end{document}